\documentclass[11pt]{article}

\usepackage{graphicx} 
\usepackage{mathtools}
\usepackage{amssymb}
\usepackage{amsthm}
\usepackage{thmtools}
\usepackage{thm-restate}
\usepackage{hyperref}
\usepackage{algorithm2e}[ruled]
\usepackage{fancyvrb}
\usepackage{xcolor}
\usepackage{multirow}
\usepackage{makecell}
\usepackage{subcaption}
\usepackage[listings,skins,breakable]{tcolorbox}
\usepackage{xfrac}
\usepackage{fullpage}
\usepackage{authblk}
\usepackage{natbib}
\usepackage{booktabs}

\def\calA{\mathcal{A}}
\def\calB{\mathcal{B}}
\def\calC{\mathcal{C}}

\def\calM{\mathcal{M}}

\def\calS{\mathcal{S}}

\def\calV{\mathcal{V}}

\def\calX{\mathcal{X}}
\def\calY{\mathcal{Y}}

\def\E{\mathbb{E}}
\def\R{\mathbb{R}}

\def\tF{\tilde{F}}

\def\tH{\tilde{H}}

\def\tK{\tilde{K}}
\def\tL{\tilde{L}}

\renewcommand{\epsilon}{\varepsilon}

\declaretheorem[name=Theorem, style=plain, numberwithin=section]{thm}
\declaretheorem[name=Definition, numberwithin=section]{defn}
\declaretheorem[name=Lemma, sibling=thm]{lemma}
\declaretheorem[name=Fact, sibling=thm]{fact}

\def\Lin{\textsf{PrivEMDLinear}}
\def\Hist{\textsf{PrivEMDItemWise}}

\def\ReductCentral{\textsf{BoundedEMDReduction}}
\def\FreqEstLocal{\textsf{FreqEstLocal}}
\def\emd{d_{\textsf{EM}}}

\def\krr{\textsf{GKRR}}

\usepackage[normalem]{ulem}
\newcommand{\add}[1]{\textcolor{black}{#1}}

\begin{document}
\title{Metric Differential Privacy at the User-Level \add{ via the Earth Mover's Distance}}


\author[1]{Jacob Imola}
\affil[1]{University of Copenhagen}

\author[2]{Amrita Roy Chowdhury}
\affil[2]{University of Michigan, Ann Arbor}

\author[3]{Kamalika Chaudhuri}
\affil[3]{University of California, San Diego}

\maketitle

\begin{abstract}
Metric differential privacy (DP) provides heterogeneous privacy guarantees based on a distance between the pair of inputs. It is a widely popular notion of privacy  since it captures the natural privacy semantics for many applications (such as, for location data) and results in better utility than standard DP. However, prior work in metric DP has primarily focused on the \textit{item-level} setting where every user only reports a single data item. A more realistic setting is that of user-level DP where each user contributes multiple items and privacy is then desired at the granularity of the user's \textit{entire} contribution. \add{In this paper, we initiate the study of one natural definition of metric DP at the user-level. Specifically,  we use the earth-mover's distance ($d_\textsf{EM}$) as our metric to obtain a notion of privacy as it captures both the magnitude and spatial aspects of changes in a user's data.}

We make three main technical contributions. First, we design two novel mechanisms under $d_\textsf{EM}$-DP to answer linear queries and item-wise queries. Specifically, our analysis for the latter involves a generalization of the privacy amplification by shuffling result which may be of independent interest. Second, we provide a black-box reduction from the general unbounded to bounded $d_\textsf{EM}$-DP (size of the dataset is fixed and public) with a novel sampling based mechanism. Third, we show that our proposed mechanisms can provably provide improved utility over user-level DP, for certain types of linear queries and frequency estimation. 
\end{abstract}

\setcounter{section}{0}
\section{Introduction}\label{sec:intro}
Differential privacy (DP) is the state-of-the art technique that enables useful data analysis while still providing a strong privacy guarantee at the granularity of individuals~\citep{dwork2006differential}. Over nearly two decades, DP has enjoyed significant academic attention and has proven its efficacy in practical applications as well. It has been successfully deployed in diverse settings, including the US census~\citep{abowd2018us}, Apple's iOS platform~\citep{cormode2018privacy}, and Google Chrome~\citep{erlingsson2014rappor}.

Intuitively, DP guarantee makes a pair of input data to be indistinguishable from each other. The standard DP guarantee requires \textit{all} pairs of inputs to be indistinguishable thereby providing a uniform privacy guarantee to all pairs. This implies that every pair of input is considered equally sensitive. However, many practical applications call for a more tailored privacy semantics based on the heterogeneity of the data. In particular, input pairs that are closer or more similar to each other are considered to be more sensitive. For instance,  for location data, revealing the exact city of residence is far more sensitive than revealing just the country. Metric DP ($d_\calX$-DP; \cite{chatzikokolakis2013broadening}) is a notion of DP that formally captures this heterogeneity in privacy semantics. Specifically, similarity is measured via a distance metric $d_\calX$ and the privacy guarantee degrades linearly 
with the $d_\calX$ distance between the pair of inputs. In addition to offering a more nuanced privacy definition, metric DP also improves utility compared to standard DP.  
This improvement stems from metric DP requiring only similar pairs of input to be indistinguishable, which results in a significantly lower noise than standard DP.


Prior work in metric DP has primarily focused on the \textit{item-level} setting where every user only reports a single data item (for e.g., a single record in a dataset). However, in many practical applications, a user contributes multiple items to a dataset. Privacy is then desired at the granularity of the user's \textit{entire} contribution. This has spurred a large body of work known as \emph{user-level} DP~\citep{amin2019bounding, bassily2023user,cummings2022mean,acharya2023discrete}. However, all of this work considers only standard DP  and is thus susceptible to the same limitations in utility as noted earlier.  \add{To this end,  we initiate the study of one natural definition of metric DP at the user-level.  In particular, we look into a specific distance metric called \emph{earth-mover's distance} ($\emd$;~\cite{givens1984class}). $\emd$ measures the similarity of two distributions and is quantified by the cost of transforming one distribution to another where the cost function can be defined by \textit{any} metric over the underlying data. Thus, $\emd$ provides a general and naturally interpretable way of measuring similarity, suitable for capturing the privacy semantics of various real-world scenarios.  While there have been some prior attempts at this, these works are limited to specific settings, such as text data~\cite{fernandes2019generalised}. To the best of our knowledge, this is the \textit{first work to propose a definition of metric DP at the user-level for a general setting via $\emd$.}}

\add{The immediate question when applying metric DP at the user level is how to define a metric on the entire collection of a user's data. We argue that $\emd$ is particularly well-suited for this task.} Recall that metric DP caters to the privacy semantics that similar data is more sensitive. But the challenge here is that the similarity between two collections (sets) of data points has to be measured along two dimensions -- $(1)$ the distance between the individual data items, and $(2)$ the fraction of the data items in the set that are different. In particular, note that in addition to small changes in the item-wise distances, changes in a \textit{smaller amount} of the data also indicate more similarity and hence, correspond to more sensitive information (see below for concrete examples). This necessitates a measure that can express both of these quantities as a single metric, for which $\emd$ is a natural choice. Informally, the $\emd$ between two distributions is the minimum cost of transporting one distribution to another, where the cost is determined by the quantity of data items moved multiplied by the distance (measured via $d_{\calX}$) over which they are moved. Our resulting privacy definition, denoted as $\emd$-DP, yields the following privacy semantics. Under $\emd$-DP, the strength of the privacy guarantee (indistinguishability) between two pairs of inputs $K,K'$ (sets of data items) grows inversely with $\tau q$ if $K'$ can be obtained by changing $\tau$ fraction of $K$ by an average distance of $q$ (Def. \ref{def:emd-dp}). $\emd$ therefore takes into account both the structure of the distributions as well as the raw difference in their values. Consequently, the parameters $\tau$ and $q$ provide flexibility in interpretation and offer a nuanced privacy definition suitable for many practical applications. We illustrate this as follows:

\textbf{Location Data.} We will use our location dataset as a canonical example throughout the paper. Suppose that the location dataset consists of daily locations of users collected over a period of time. Here, the parameter $\tau$ can be interpreted in terms of the  length of the time window the change in $K'$ pertains to, and $q$ corresponds to the extent of change in the location. Then, $\emd$-DP  makes it harder to distinguish between locations that are $(1)$ close to each other, and $(2)$ collected over a smaller time window. This is natural, since locations gathered over an extended period, such as a month, may reveal routine patterns that are less sensitive than locations recorded on a single day (for instance, a single-day location might reveal a non-routine visit to a friend or hospital).

\textbf{Textual Data.} Consider a natural language dataset of user conversations where each user's data is represented as a set of words. Typically, word embeddings $\phi$ map each word into a high-dimensional space, and word similarity is measured using a distance, such as the Euclidean distance, between $\phi(x_1)$ and $\phi(x_2)$. Now, the parameter $\tau$ corresponds to what fraction of the user's conversation has changed in $K'$ from $K$, while $q$ corresponds to the extent of the changes in the textual content. Thus, two conversations are harder to distinguish if $(1)$ there is only a fine-grained difference in their textual semantics\footnote{Such as transitioning from text about algebra to trigonometry versus changing it from ``math'' to ``classical music''.}, and $(2)$ if it pertains to just a small fraction of the conversation (indicating a user rarely discussed the topic, which typically implies more sensitive information). 

\textbf{Graph Data.} Consider a graph $G = (V, E)$ in which connections in $E$ are private. Suppose there is additional public information in the form of a covariate $\phi : V \rightarrow \R^d$, which captures some auxiliary information about a user---for instance, the \emph{interests} of a user. Here similarity between users is measured via covariate distance. The parameter $\tau$ corresponds to the fraction of a user's connections which has changed in $K'$ from $K$, and the parameter $q$ corresponds to the extent of the change in their interests. Thus, two graphs are harder to distinguish between if $(1)$ it is a fine-grained change to the interest\footnote{for instance, shifting from movies featuring Dwayne Johnson to Vin Diesel  instead  of from ``action'' to ``rom-com''}, and $(2)$ if it pertains to only a few of the user's connections (say a small, private group of friends). This again captures natural privacy semantics as users are more likely to share common interests with their close friends than with a larger group, such as all workplace colleagues.


\add{In a nutshell, $\emd$-DP offers a more fine-grained privacy definition compared to standard DP, that captures real-world privacy semantics more effectively while providing better utility. For example, consider the following two instances of $K'$ in the aforementioned example of location data  -- one where the data for an entire month is different and another where only a single day’s data differs. Standard DP treats both cases as equally sensitive (i.e., offers the same privacy guarantee for both cases), necessitating a larger noise addition even in the latter case, which results in reduced utility. In contrast,  $\emd$-DP offers a stronger privacy guarantee for the latter, thereby resulting in a better privacy-utility trade-off. $\emd$-DP is in fact a natural relaxation of standard DP. 
We provide more details on the interpretation of $\emd$-DP relative to standard DP in Sec.~\ref{sec:interpret}. }
A full version of this paper appears in~\cite{imola2024metric}.


\subsection{Details of Our Contributions}\label{sec:intro:contributions}

We consider $n$ users who hold datasets $\{K_i\}_{i=1}^n$, each containing elements from a \emph{data domain $\calX$} of size $k=|\calX|$. Let $d_\calX$ denote a distance metric defined over $\calX$. WLOG, we consider $d_\calX$ to be normalized such that all measures of distance are at most $1$. Furthermore, we will let $\tK_i$ denote the dataset $K_i$, where each point in $K_i$ has a weight of $\frac{1}{|K_i|}$.  The $\emd$ between any pair of datasets $\{K_i, K_i'\}$ can be defined by first representing them in $\{\tilde{K}_i, \tilde{K}_i'\}$, and then using $d_\calX$ to measure the minimum cost of transporting $\tilde{K}_i$ to $\tilde{K}_i'$, where cost measures the sum of all distances moved by each point, weighted by the weight of that point. The global dataset is given by  $K_{G} = K_1 \cup \cdots \cup K_n$,
and there is an aggregator who wants to privately compute a query $V(K)$. In the \textbf{central model}, the aggregator already holds $K_i$ from each user, and applies a private mechanism $\calM(K_G)$ to obtain a private estimate for $V$. In the \textbf{local model}, the users do not trust the aggregator, and communicate private messages $\{m_i = \calM_i(K_i)\}$ to the aggregator. The aggregator then post-processes these messages $\calV(m_1, \ldots, m_n)$ to output a private estimate of $V$. For simplicity, in this work we assume the mechanisms $\calM_i$ to be non-interactive.

We also make a distinction between \emph{bounded} and \emph{unbounded} data. Note that boundedness here refers to the size of each user's dataset and \textit{not} the number of the users -- throughout the paper, we assume that the number of users, $n$, is fixed and publicly known. In our specific context, bounded data corresponds to the case where the size of each user's dataset is publicly known, and the mechanism $\calM$ only needs to preserve privacy between datasets of the same size. Furthermore, in the central model, each user's dataset has the \textit{same} public size. The benefit of this simplification is that algorithm analysis is easier. Such a bounded data setting has been considered in many previous works (see ~\cite{bookDP}). We also consider the general unbounded data setting where each user can have datasets of varying sizes, with the size being \textit{private} as well.

For each model and type of boundedness, we summarize how one would apply $\emd$-DP, along with the resulting semantics, in Table~\ref{tab:privacy}. We also include a corresponding notion of the standard user-level DP~\cite{liu2023algorithms} (provides a uniform privacy guarantee to all pairs of datasets) which serves as our baseline. In what follows, we elaborate on our main contributions.
\begin{table*}
\scalebox{0.7}{\begin{tabular}{ccccp{5cm}p{5cm}} \toprule
Model & Granularity & Data Boundedness & Privacy Guarantee & Semantics & Notes \\ \midrule
\multirow{15}{2cm}{Local (applies to each $\calM_i$)} & User& Unbounded & \makecell{\\$(\epsilon, \delta)$-user-level DP \\ (Def. \ref{def:user-level-dp-local})} & Two input datasets $K,K' \in \calX^*$ are indistinguishable with parameters $(\epsilon, \delta)$ & Recently proposed in~\cite{acharya2023discrete}. Acts our baseline for the local model. \\
& User & Bounded & \makecell{\\$(\alpha, \delta)$-bounded $\emd$-DP \\ (Def.~\ref{def:emd-dp})} & Two input datasets $K,K' \in \calX^m$ are indistinguishable with parameters $(\alpha \emd(\tilde{K},\tilde{K}'), \delta)$. & The size of each dataset, $m$, is public. Proofs of privacy easier due to Lemma~\ref{lem:bvn}.\\
& User & Unbounded & \makecell{\\$(\alpha, \delta)$-unbounded $\emd$-DP \\ (Def.~\ref{def:emd-dp})} & Two input datasets $K,K' \in \calX^*$ are indistinguishable with parameters $(\alpha \emd(\tilde{K},\tilde{K}'), \delta)$. & Implies user-level DP when $\alpha \leq \epsilon$ since $\emd(\cdot,\cdot)\leq 1$. \\
& User & Unbounded & \makecell{\\$(\epsilon, \delta, r)$-discrete $\emd$-DP \\ (Def.~\ref{def:discrete-emd-dp-local})} & Two input datasets $K, K' \in \calX^*$ such that $\emd(\tK, \tK') \leq r$ are indistinguishable with parameters $(\epsilon, \delta)$ & Using group privacy, can show $(\epsilon \lceil \frac{d}{r} \rceil, \delta \lceil \frac{d}{r} \rceil \exp(\epsilon \lceil \frac{d}{r} \rceil))$ for two $K, K'$ s.t. $\emd(\tK, \tK') \leq d$. \\
& Item & N/A & \makecell{\\$(\alpha, \delta)$-$d_\calX$-DP \\ (Def.~\ref{def:metric-dp})} & Two input items $x,x' \in \calX$ are protected with parameters $(\alpha d_\calX(x,x'),\delta)$  & Proposed in \cite{chatzikokolakis2015constructing}. \\ \midrule
\multirow{19}{2cm}{Central (applies to $\calM$)} & User & Unbounded & \makecell{\\$(\epsilon, \delta)$-user-level DP \\(Def. \ref{def:user-level-central-dp})}& Let $K_G=K_1\cup \cdots K_n$ where $K_i \in \calX^*$. Two input global datasets $K_G,K_G'$ s.t. they differ only on the dataset of a single user $\{K_i,K_i'\}, i \in [n]$ are indistinguishable with parameters $(\epsilon, \delta)$ & Studied widely \cite{ bassily2023user,liu2020learning,liu2023algorithms}. Acts our baseline for the central model.\\
& User & Bounded & \makecell{\\$(\alpha, \delta)$-bounded $\emd$-DP \\ (Def.~\ref{def:emd-dp-central})} & Two input global datasets $K_G,K_G'$ s.t. they differ only on $\{K_i,K_i'\} \in \calX^m\times \calX^m$ are indistinguishable with parameters $(\alpha \emd(\tilde{K}_i,\tilde{K}_i'), \delta)$ & Each $K_i$ has size $m$ which is public. \\
 & User & Unbounded & \makecell{\\ $(\epsilon, \delta, r)$-discrete $\emd$-DP \\ (Def.~\ref{def:discrete-emd-dp-central})} & Two input global datasets $K_G,K_G'$ s.t. they differ only on $\{K_i,K_i'\}$ and $\emd(\tilde{K}_i,\tilde{K}_i')\leq r$ are indistinguishable with parameters $(\epsilon, \delta)$. & Using group privacy, can show parameters $(\epsilon \lceil \frac{d}{r} \rceil, \delta \exp(\epsilon \lceil \frac{d}{r} \rceil))$ for any $K_i, K_i'$ s.t. $\emd(\tilde{K}_i,\tilde{K}_i') \leq d$. Implies user-level DP when $r \geq 1$.\\
 \bottomrule
\end{tabular}}

\caption{Summary of privacy definitions for this paper.  The number of users, $n$, is fixed and publicly known for the definitions in the central model. 
}
\label{tab:privacy}
\end{table*}

\subsubsection{Mechanism Design} 
We provide novel mechanisms for answering two types of queries for $\emd$-DP. \add{In this section, we let $K$ denote a general dataset of interest, which in the local model would be set to $K_i$ and in the global model would be set to $K_G$.}

\subsubsection*{Linear Queries} First, we study how to release linear queries $F\tilde{K}$, where $F \in \R^{d\times |\calX|}$ is a real-valued matrix with bounded entries. \add{While computing the global sensitivity of such a query is easy under user-level DP, proving a sensitivity bound under $\emd$-DP is considerably more complex. It involves making a stronger Lipschitz assumption about the points in $F$, and then leveraging this property to transport $\tK$ onto $\tK'$ to compute an upper bound on how much $F\tK$ can change by.} To this end, we first prove the following bound:
\begin{thm}
    (Informal version of Theorem~\ref{thm:sens}): The sensitivity of $F\tK$ is upper bounded by
    \[
        \max_{K, K'} \frac{\|F\tilde{K} - F\tilde{K}'\|}{\emd(\tilde{K}, {\tilde{K}}')} \leq \max_{x, x' \in \calX} \frac{\|F_x - F_{x'}\|}{d_\calX(x, x')},
    \]
    where the notation $F_x$ indicates the column of $F$ indexed at $x$.
\end{thm}
Using the above result, we show that the sensitivity of $F$, which is a maximum over the space of all datasets, can be reduced to the Lipschitzness of $F$, which is simpler to bound. \add{Sensitivity analysis for standard DP typically only requires 
$F$ to be bounded. However, in the case of $\emd$-DP, we need a stronger assumption -- 
$F$ must also be Lipschitz. In Section~\ref{sec:lin-avg}, we demonstrate that several commonly used queries, such as average, region, and similarity queries, are indeed Lipschitz.} 

\subsubsection*{Unordered Release of Item-wise Queries} We design a mechanism for performing itemwise queries on the \textit{entire} dataset $K$. \add{For now we consider bounded DP, so $|K|$ is known in advance. Our approach is simple -- apply a private mechanism $\calA$ to each item $x_j \in K$ and then release the set of noisy outputs $\{\calA(x_j)\}$ after shuffling them}. Here $\calA$ can be an arbitrary mechanism satisfying $(\alpha,0)$-$d_\calX$ DP which makes our mechanism completely general-purpose (see Section \ref{sec:hist-release} for some concrete examples of $\calA$). 
\add{The main technical novelty lies in the privacy analysis of the above mechanism. While one can use composition to show this release satisfies $O(\alpha m)$-$\emd$ DP~\cite{fernandes2019generalised}, this is far from being tight. As discussed in Section~\ref{sec:hist-release}, composition is not the right tool for tight privacy analysis since it does not account for the fact that the output of our mechanism is an unordered list, i.e., the $\calA(x_j)$'s are released in a random order. 
Instead, we leverage privacy amplification by shuffling~\citep{feldman2022hiding} to prove a general shuffling result and adapt it for $\emd$-DP. }
\begin{thm}\label{thm:shuffle-metric-full-inf} (Informal version of Theorem~\ref{thm:shuffle-metric-full})
    Suppose that $\calA : \calX \rightarrow \calY$ is an $\alpha$-$d_\calX$ DP algorithm with respect to $d_\calX$. Let $(x_1, \ldots, x_m) \in \calX^m$ be a dataset. Then, releasing $\mathsf{Shuffle}(\calA(x_1), \ldots, \calA(x_m))$ satisfies $(O(\alpha\sqrt{m e^{\alpha} \ln (m / \delta)}), \delta e^{\alpha})$-$\emd$-DP.
\end{thm}
This analysis reduces the cost of releasing $m$ points in the multiset from $m \alpha$ to $\sqrt{m} \alpha$, allowing for better utility. We keep the analysis general -- we consider releasing the shuffled multiset of any black-box mechanism $\calA$, that satisfies metric DP in the data domain $\calX$, applied to each data point. Consequently, this result has broader applications to the shuffle model of privacy, and may be of independent interest.

\subsubsection{Extending $\emd$-DP to the Unbounded Setting}
We start our mechanism designs by considering the bounded data setting in both the local and central models of privacy (see Table~\ref{tab:privacy}) as this enables easier privacy analysis (Section \ref{sec:mechanisms}). 
However, the bounded setting might be restrictive in practice as it cannot support usecases where users have different amounts of data, or the data sizes are also private. To this end, we extend $\emd$-DP to the more general unbounded setting. \add{We achieve this through a general reduction that converts any bounded $\emd$-DP mechanism into an unbounded one while treating the mechanism as a black box. If the data from each user is relatively homogeneous, such as being i.i.d., then the utility of the mechanism will be preserved.}

\add{Our reduction uses a projection mechanism that projects any dataset onto one with a fixed size, without significantly increasing the $\emd$ distance. The projection we use is \emph{sampling with replacement}. Intuitively, this is a smooth projection because we can view sampling from two datasets $K, K'$ in terms of a coupling between them, and show that the sampled points can also be coupled with expected cost given by $\emd(K, K')$. Using Bernstein’s inequality, we show convergence to $\emd(K, K')$, up to a small additive factor.}

One caveat is that the introduced additive factor necessitates a slight adjustment to the privacy semantics of $\emd$-DP. \add{Whereas $\emd$-DP protects a change of $\emd$-distance $d$ with an effective privacy parameter $\alpha d$, for all $d > 0$, our modified definition includes $r$ as a fixed parameter, and states that all changes of $\emd \leq r$ are protected uniformly with parameter $\alpha r$. We may view this as a discretization of the $\emd$ by rounding it up to the nearest multiple of $r$.  We refer to this notion as $(\alpha r, \delta, r)$-\emph{discrete} user-level $\emd$-DP (Def. \ref{def:discrete-emd-dp-local}). This privacy guarantee is weaker than $\emd$-DP only for changes of $\emd < r$---for bigger changes, the two definitions are equivalent up to factors of $2$ using group privacy. Typically, $r$ is small and the practical difference between the definitions is negligible. Our reduction satisfies the following:}

\begin{thm}\label{thm:emd-dp-central-inf} (Informal version of Theorem~\ref{thm:emd-dp-central})
    Suppose that for $n$ users, $\calM$ is a mechanism which satisfies $(\alpha, \delta)$-bounded $\emd$-DP. The algorithm which, given arbitrary user datasets $K_1, \ldots, K_n$, takes $s$ i.i.d. samples from each $K_i$ and then applies $\calM$ on each of the sampled data items, satisfies $(\alpha r, \delta, r)$-discrete $\emd$-DP (in the central model) for all $r \geq \frac{2 \ln (1/\delta)}{s}$.
\end{thm}
The two notions of privacy are nearly equivalent for small $r$, showing that unbounded $\emd$-DP can be reduced to bounded $\emd$-DP with an almost exact translation of the privacy guarantee.

\subsubsection{Demonstrating Improvements Over User-level DP} \add{Finally, we evaluate the benefit of using the more nuanced privacy semantics of $\emd$-DP over standard user-level DP by comparing the privacy and utility of our proposed mechanisms with baselines.} Specifically, in Sec.~\ref{sec:mean-est}, we study a special type of linear query called \emph{linear embedding queries} and in Sec.~\ref{sec:freq-est}, we study private frequency estimation. For simplicity, we consider the bounded data setting.

Let's start by understanding the relationship between $(\alpha, \delta)$-$\emd$-DP and $(\epsilon, \delta)$-user-level DP. The following observations hold in both the central and local models:
\begin{itemize}
    \item $\alpha = \epsilon$: Since we assume $d_\calX$ is normalized, we always have $\emd \leq 1$. Thus, in this case $(\epsilon, \delta)$-$\emd$-DP implies $(\epsilon, \delta)$-user-level DP. However, any pair of input $K,K'$ such that $\emd(\tilde{K},\tilde{K}')
< 1$ the privacy protection of $\emd$-DP is actually stronger. (Note that if $\alpha \leq \epsilon$, then user-level DP is strictly weaker than $\emd$-DP; the more appropriate baseline is to use $\alpha = \epsilon$.)

    \item $\alpha > \epsilon$: In this case, some pairs of inputs (with a large $\emd$ distance between them) are protected less strongly than they are under user-level DP. However, as indicated in our aforementioned real-life examples, input pairs with high $\emd$ (i.e., dissimilar input pairs) are typically less sensitive.
\end{itemize}

Now, we interpret the theoretical error bounds for linear embedding queries.
From Table~\ref{tab:lin-query-error},  the error for releasing a $d$-dimensional linear embedding query under user-level DP is $O(\frac{d}{\epsilon n})$, while it is $O(\frac{d}{\alpha n})$ for $\emd$-DP. When $\alpha = \epsilon$, these utilities are identical, but $\emd$-DP offers stronger privacy. When $\alpha > \epsilon$, then the utility of $\emd$-DP is higher than that of user-level DP, with the the two guarantees offering differing privacy semantics. Thus, in both cases, there is a clear benefit of using $\emd$-DP. These observations are the same in the local model.

Finally, for frequency estimation in the local model, Table~\ref{tab:hist-error} shows that the error of user-level DP is $O(\sqrt{\frac{k^2 \ln(m/\delta)}{n\epsilon^2}})$, while it is $O(\sqrt{ \frac{k^3}{n\alpha^2} \max \{ \ln(\frac{m}{\delta}), \alpha\}}$ for $\emd$-DP. For constant $\epsilon$ and $\alpha \geq \epsilon^2 \frac{k}{\ln(m / \delta)}$, the utility is improved. In the central model, the error of the user-level DP algorithm is $O(\frac{k}{n\epsilon})$ while it is $O(\frac{k^{3/2}}{n\alpha} \sqrt{\max\{\ln(\frac{m}{\delta}), \alpha \}})$ for $\emd$-DP. The algorithm under $\emd$-DP has the added benefit that it can be implemented in the shuffle model of privacy, which requires less trust and parallels prior work in the shuffle model~\cite{feldman2022hiding}. There is a utility improvement for $\alpha \geq \epsilon^2 k$. When $\epsilon \leq \alpha \leq \epsilon^2 k$, we leave it as an interesting open problem whether $\emd$-DP can offer utility improvements over user-level DP.

\begin{table*}
\begin{subtable}[t]{0.9\textwidth}
\setlength{\tabcolsep}{10pt} 
\renewcommand{\arraystretch}{1.5} 
\scalebox{0.65}{\begin{tabular}{c|c|c|c c |c}
    \multicolumn{6}{c}{Linear Embedding Queries} \\ \toprule 
    Algorithm & Privacy Guarantee & Privacy Model & \multicolumn{2}{c}{$\ell_2$ Error} & Notes \\ \midrule 
    $K$-norm Mechanism & $(\epsilon, 0)$-user level DP$^\flat$ & Central, Bounded & $O(\frac d {\epsilon n})$ & (Lemma~\ref{lem:lin-user-cent}) & \multirow{2}{7cm}{$\emd$-DP gives same utility but stronger privacy for $\alpha=\epsilon$; $\emd$-DP gives better utility but different privacy for $\alpha> \epsilon$. } \\
    \Lin{} & $(\alpha, \delta)$-$\emd$-DP & Central, Bounded & $O(\frac {d } {\alpha n}\sqrt{\ln\frac{1}{\delta}})$ & (Lemma~\ref{lem:lq-local}) & \\ \bottomrule \multicolumn{6}{c}{$^\flat$For random linear queries, this algorithm also has best-known error among all $(\epsilon, \delta)$-user-level DP algorithms.}
\end{tabular}}
\caption{Comparison of $\emd$-DP to user-level DP in the central model for releasing a $d$-dimensional linear embedding query. The errors in the local model are a factor $\sqrt{n}$ higher. }\label{tab:lin-query-error}
\end{subtable}
\begin{subtable}[t]{0.9\textwidth}
\setlength{\tabcolsep}{10pt} 
\renewcommand{\arraystretch}{1.5} 
\scalebox{0.63}{\begin{tabular}{c|c|c|c c|c}
    \multicolumn{6}{c}{Frequency Estimation} \\ \toprule 
    Algorithm & Privacy Guarantee & Privacy Model & \multicolumn{2}{c}{$\emd$ Error} & Notes \\ \midrule 
    Hadamard Response & $(\epsilon, 0)$-user-level DP$^\flat$ & Local, Bounded & $O\left(\sqrt{\frac{k^2 \ln (m/\delta)}{n\epsilon^2}}\right)$ & (Lemma~\ref{lem:local-hist-utility-baseline}) & \multirow{2}{5cm}{Assuming $k, \epsilon, \alpha \leq \sqrt{m}$; $\emd$-DP gives better utility for $\alpha \geq \epsilon^2 \frac{k}{\ln(m / \delta)}$.} \\
        \Hist{} & $(\alpha, \delta)$-$\emd$-DP & Local, Bounded & $O\left( \sqrt{\frac{k^3}{n\alpha^2}\max\left\{\ln(\frac m \delta), \alpha \right\}}\right)$ & (Thm.~\ref{thm:krr-utility}) &  \\ \hline
    Laplace Mechanism & $(\epsilon, 0)$-user-level DP$^\flat$ & Central, Bounded & $O\left(\frac{k}{n \epsilon} \right)$ & (Lemma~\ref{lem:lap-freq}) & \multirow{2}{5cm}{Assuming $n < \frac m \alpha$; $\emd$-DP gives better utility when $\alpha >\epsilon^2 k $ .} \\
    \Hist{} & $(\alpha, \delta)$-$\emd$-DP & Central, Bounded & $O\left(\frac{k^{3/2}}{n \alpha} \sqrt{\max \left\{\ln(\frac m \delta), \alpha\right\}}\right)$ & (Corollary~\ref{coro:krr-utility-central}) & \\ \bottomrule 
    \multicolumn{6}{c}{$^\sharp$This algorithm works in the shuffle model, which requires less trust than the central model.} \\ 
    \multicolumn{6}{c}{$^\flat$This algorithm also has best-known error among all $(\epsilon, \delta)$-user-level DP algorithms.}
\end{tabular}}
\caption{Comparison of $\emd$-DP to user-level DP for frequency estimation in the setting defined in Section ~\ref{sec:freq-est}. $k$ is the domain size $|\calX|$. }\label{tab:hist-error}
\end{subtable}
\caption{Summary of theoretical utility guarantees, assuming there are $n$ users who hold datasets of size $m$. 
}\label{tab:summary}
\end{table*}

\section{Background}\label{sec:prelims}
\subsection{Differential Privacy}
Intuitively, DP is a property of a mechanism which ensures that its output distribution remains insensitive to changes in the data of a single individual. The standard DP guarantee, which is also know as \textit{item-level} DP, considers each user $U_i$ to contribute only a single item $x_i \in \calX$ to a global dataset, i.e.,  $K_i=x_i$. \add{Instead we consider user-level DP, where each $K_i$ is itself a multiset of elements from $\calX$. We denote the set of all multisets of elements from $\calX$ as $\calX^*$. In the local model, our privacy definition is then:}
\begin{defn}[Unbounded User-level Local DP~\cite{acharya2023discrete}]\label{def:user-level-dp-local}
   We say a mechanism $\calM$ acting on a dataset $K$ satisfies $(\epsilon, \delta)$-unbounded user-level local DP if, for all $K, K' \in \calX^*$ and all sets of outputs $O$, we have
    \begin{equation}\label{eq:dp}
        \Pr[\calM(K) \in O] \leq e^{\epsilon} \Pr[\calM(K') \in O] + \delta.
    \end{equation}
\end{defn}
Note that here we consider the more general unbounded data setting where the two datasets $\{K,K'\}$ can have arbitrary sizes.

Next, we present the definition for the central model.
\begin{defn}[Unbounded User-level Central DP~\cite{liu2023algorithms}]\label{def:user-level-central-dp}  Let $K_G=K_1\cup\cdots\cup K_n$ denote a global dataset from $n$ users where $\forall i \in [n], K_i\in \calX^*$. We say $K_G \sim K'_G$, if $K'_G$ can be obtained from $K_G$ by changing the dataset of a single user $U_i$ from $K_i$ to $K_i'$. 
   We say a mechanism $\calM$ acting on a dataset $K$ satisfies $(\epsilon, \delta)$-unbounded user-level central DP if, for all $K_G, K_G'$ such that $K_G\sim K'_G$, and all sets of outputs $O$, we have
    \begin{equation}
        \Pr[\calM(K_G) \in O] \leq e^{\epsilon} \Pr[\calM(K_G') \in O] + \delta.
    \end{equation}
\end{defn}
Note that there is no restriction on the sizes of the datasets $\{K_i\}, i \in [n]$ in the above definition.


 Next, we  define metric DP that enables the privacy guarantee to depend on a metric $d_\calX$ between the pair of inputs.  We start by introducing it at the item-level (i.e., we consider changing only one item $x \in \calX$ to another item $x' \in \calX$). For simplicity, we consider the local model, so the mechanism acts on just a single item:

\begin{defn}[Local $d_\calX$-DP~\cite{Alvim2018}]\label{def:metric-dp}
     We say $\calM$ satisfies $(\alpha, \delta)$-local $d_\calX$-DP if for all data elements $x, x' \in \calX$, and all sets of outputs $O$, we have
     \[
        \Pr[\calM(x) \in O] \leq e^{\alpha d_\calX(x, x')} \Pr[\calM(x') \in O] + \delta.
     \]
\end{defn}
We replace the traditional privacy parameter $\epsilon$ with $\alpha$ in the above definition, because $\epsilon$ in Definitions.~\ref{def:user-level-dp-local} and \ref{def:user-level-central-dp} is a unitless  parameter while $\alpha$ has the inverse unit of $d_\calX$.

\subsection{Earth-Mover's Distance} 
\textbf{Notations.} We view datasets as multisets of elements from $\calX$. 
We will also view a dataset $K \in \calX^*$ as a \emph{probability distribution} defined by its normalized histogram $\tK$. To do so, let $\Delta^\calX \subseteq \R^\calX$ denote the probability simplex indexed by $\calX$---i.e. the set of all vectors $\langle v_x \rangle_{x \in \calX}$ such that $v_x \geq 0$ and $\sum_{x \in \calX} v_x = 1$. For a dataset $K$, $\tK \in \Delta^\calX$ then denotes the probability distribution defined by $K$, meaning $\tK[x] = \frac{\text{Num. occurrences of $x$ in $K$}}{|K|}
$. The earth-mover's (or $1$-Wasserstein) distance~\cite{givens1984class} is defined as follows. For a joint distribution $C(x_1, x_2) \in \Delta^{\calX \times \calX}$, let $C_{x_1}(x_2)$ denote the distribution conditioned on observing $x_1$, and let $C_1(x_1)$ denote the marginal distribution of $x_1$. We define $C_{x_2}(x_1)$ and $C_2(x_2)$ similarly.

\begin{defn}\label{def:coupling}
    For distributions $P, Q \in \Delta^\calX$, a joint distribution $C$ on $\calX \times \calX$ is a coupling between $P$ and $Q$ if $C_1 = P$ and $C_2 = Q$. We let $\calC(P, Q)$ denote the set of couplings between $P$ and $Q$. 
\end{defn}

A coupling $C$ can be viewed as a ``transportation plan'' between $P$ and $Q$, in the sense that if $C$ places $m$ probability mass at a point $(x_1, x_2)$, then $m$ probability mass from $P$ at $x_1$ is transported to $Q$ at $x_2$ (or vice-versa). We define the \emph{cost} of a coupling as the expected transportation distance given by $\E_{(x,x') \sim C} d_\calX(x, x')$. The \emph{earth-mover's distance} ($\emd$) between $P, Q$ is equal to the minimum possible cost of a coupling between $P$ and $Q$:
    \[
        \emd(P, Q) = \inf_{C \in \calC(P, Q)} \operatorname{\mathbb{E}}_{(x, x') \sim C} d_\calX(x, x').
    \]
Since we assume that $d_{\calX}$ is bounded by $1$, we  have $\emd(\cdot, \cdot) \leq 1$.

Next, we present the Birkhoff-Von Neumann theorem which is useful in our privacy analysis in Section~\ref{sec:hist-release}. The theorem states that if both $P$ and $Q$ are empirical distributions with the same number of points, then the $\emd$ between them is the cost of the coupling that moves the entire mass in each point to the same destination:

\begin{lemma}[Birkhoff-Von Neumann Theorem~\cite{konig2001theorie}, Lemma A.1 in~\cite{fernandes2019generalised}]\label{lem:bvn}
    For two datasets $K = \{x_1, \ldots, x_m\}$ and $K' = \{y_1, \ldots, y_m\}$, there is a permutation $\pi : [m] \rightarrow [m]$ such that
    \begin{equation}\label{eq:bvn}
        \emd(\tK, \tK') = \frac{1}{m}\sum_{i=1}^m d_\calX(x_i, y_{\pi(i)}).
    \end{equation}
\end{lemma}

\section{Definition of $\emd$-DP}\label{sec:semantics}
In this section, we introduce our generalization of metric DP to the user-level. We start with the local model. We use the $\emd$ metric to measure the distance between two datasets $K, K'$ since it captures the intuition that the changes which move smaller amounts of data by smaller distances are more sensitive (as discussed in Section~\ref{sec:intro}).

\begin{defn}[(Un)Bounded Local $\emd$-DP]\label{def:emd-dp}
    Let $\calM$ be a mechanism which acts on a dataset $K$. We say $\calM$ satisfies $(\alpha, \delta)$-\textbf{bounded} local $\emd$-DP if for any two datasets $K, K'$ such that $|K| = |K'|$, and for all sets of outputs $O$, we have
    \begin{equation}\label{eq:metric-dp}
        \Pr[\calM(K) \in O] \leq e^{\alpha \emd(\tK, \tK')} \Pr[\calM(K') \in O] + \delta.
    \end{equation}
    If the above equation holds for all datasets $K, K'$, regardless of whether $|K| = |K'|$, we say that $\calM$ satisfies $(\alpha,\delta)$-\textbf{unbounded} local $\emd$-DP.
\end{defn}

For bounded $\emd$-DP, the size of the dataset is \emph{not} protected, which is acceptable for  applications where the amount of data is not sensitive. We explicitly differentiate between bounded and unbounded data since privacy analysis is easier under bounded $\emd$-DP by leveraging Lemma~\ref{lem:bvn} (see Section~\ref{sec:mechanisms}). \\
In the central model, our goal is to protect changes in a single user's dataset, transitioning from $K_i$ to $K'_i$, with a privacy guarantee that depends on $\emd(\tK_i, \tK_i')$. We consider the bounded data setting where each dataset $K_i$ has a publicly known fixed size $m$.

\begin{defn}[Bounded Central $\emd$-DP]\label{def:emd-dp-central}
    Let $K_G=K_1\cup\cdots\cup K_n$ denote a global dataset from $n$ users where $\forall i \in [n], K_i\in \calX^m$. We say $K_G \sim K'_G$ if $K'_G$ can be obtained from $K_G$ by changing the dataset at a single index $i$ from $K_i$ to $K_i'$. We say a mechanism $\calM$ satisfies $(\alpha, \delta)$-\textbf{bounded} central $\emd$-DP if, for all $K_G, K_G'$ such that $K_G \sim K'_G$, and all sets of outputs $O$, we have
    \[
        \Pr[\calM(K_G) \in O] \leq e^{\alpha \emd(\tilde{K}_i,\tilde{K}_i')} \Pr[\calM(K_G') \in O] + \delta.
    \]
\end{defn}

In the above definition, the two global datasets $K_G, K_G'$ are indistinguishable with a privacy parameter $\alpha \emd(K_i, K_i')$.  Since we consider the bounded data setting, neither the number of total users, $n$, nor the size of the individual datasets, $m$, are protected.  


It is important to note that the above definition cannot be directly translated to the unbounded data setting. This limitation arises from the fact that if each $K_i$ is allowed to have an arbitrary size, then changing a single $K_i$ could potentially change the entirety of $K_G$ in the worst-case (where user $U_i$ contributes the entire global dataset). This essentially reduces the central model (Def.~\ref{def:emd-dp-central}) to the local model (Def.~\ref{def:emd-dp}). 
We circumvent this challenge and provide a privacy definition for the unbounded data setting in Section~\ref{sec:reduct}, by controlling the amount of data from each user.
\\\\
\noindent
\textbf{Setting the Privacy Parameters.} There are some semantic differences between the parameter $\alpha$ in Defns.~\ref{def:emd-dp} and~\ref{def:emd-dp-central}, and $\epsilon$ in Defns.~\ref{def:user-level-dp-local} and~\ref{def:user-level-central-dp}. The privacy parameter $\epsilon$ is unitless. On the other hand, $\alpha$ is not unitless -- it has a unit inversely proportional to $\emd$. While $\epsilon \gg 1$ is usually not considered acceptable for standard DP, it is not unreasonable to set $\alpha \gg 1$ in our case. This is acceptable if a strong privacy guarantee is needed only for input pairs that are close to each other since $\emd(\cdot,\cdot) < 1$. For all $q, \tau \in [0,1]$, let  $\mathcal{E}(q, \tau)$ refer to the minimum privacy parameter that is acceptable over all data changes of the form 
\begin{quote}
    A $\tau$-fraction of $K$ is changed by average distance $q$.
\end{quote}
Then, $\alpha$ may be set as
    $\alpha = \inf_{q, \tau \in [0,1]} \tfrac{\mathcal{E}(q, \tau)}{q \tau}$,
and we can verify that Defn.~\ref{def:emd-dp} will protect an input pair with the corresponding budget $\mathcal{E}(q,\tau)$. The parameter $\delta$ has the same interpretation as in standard DP, and should be set $\delta \ll \frac{1}{poly(n)}$. 
\\\\
\noindent\textbf{Concrete Example.} Throughout this paper, we consider a dataset of $n=10^5$ users, each of whom contributes $m=10^3$ location data points over the period of a month. We use the length of the shortest path on earth's surface as our metric $d_\calX$. Suppose we want to protect a user's location over any particular day within a radius of $1000$ miles, and the user's location over the entire time period within a distance of $100$ miles. In the normalized metric space, these distances are $q_1 = 0.08$ and $q_2 = 0.008$, respectively\footnote{The maximum surface distance between two points on Earth is $\approx 12930$ miles.}. They correspond to a fraction $\tau_1 = \frac{1}{30}$ and $\tau_2 = 1$ of the metric space changing, respectively. Suppose we want to protect both of these inputs with privacy parameter $\epsilon = 0.2$. Hence, we set $
    \alpha = \min \left \{\frac{\epsilon}{\tau_1 q_1}, \frac{\epsilon}{\tau_2 q_2}\right\} = 25.
$
This value is much higher than typical privacy parameters used in DP, and yet it is able to adequately protect the desired inputs. Finally, we will set $\delta = 10^{-12}$ in our examples.
\section{Mechanisms for $\emd$-DP}\label{sec:mechanisms}
Now, we describe our mechanisms for releasing queries under $\emd$-DP. Throughout this section, we focus on the bounded data setting, and consider both the local and central models by considering a general dataset $K$ which will be set to be $K_i$ in the local model and $K_G$ in the central one. In Section~\ref{sec:lin-avg}, we show how to bound the sensitivity of \emph{linear queries}, which can then be released with the addition of calibrated noise. Then, in Section~\ref{sec:hist-release}, we show that we can release a noisy representation of $\tK$ under $\emd$-DP by applying \emph{any} $d_\calX$-DP mechanism to each item in $K$, and shuffling the outputs. Full proofs for this section appear in Appendix~\ref{app:mech-proofs}.

\subsection{Linear Queries}\label{sec:lin-avg}
A non-adaptive linear query on a dataset $K$ computes the value of $F \tK$, where $F \in \R^{d \times |\calX|}$ is a matrix with $d$ rows. The \emph{linearity} comes from the linear transformation $F$; our linear queries are normalized since they operate on $\tK$ rather than $K$. Such normalized queries can be used for answering the fraction of users satisfying a predicate~\cite{blum2013learning}. Nevertheless, one can estimate the non-normalized query by multiplying by an estimate of $|K|$. 

Let us represent $F$ by a function $f : \calX \rightarrow \R^d$ where $f(x) = F[x]$, the $x$th column of $F$. The linear query can then be re-written as 
\begin{equation}\label{eq:embed-query}
q_f(K) = \E_{x \sim \tK}[f(x)].
\end{equation}
Thus, we may interpret a linear query on $\tK$ as \textit{expected} value of $f$ over a random item from $K$.
Linear queries are simple but capable of expressing many indispensible tools in data analysis, and they are well-studied in differential privacy~\citep{blum2013learning, hardt2009geometry, dwork2014algorithmic}. We will design a simple mechanism satisfying $\emd$-DP for releasing a linear query, based on bounding the sensitivity of $q_f$ under the $\emd$. The sensitivity measures the maximum change output $q_f$, measured according to some norm $\|\cdot \|$ on $\R^d$, relative to a change in the inputs by a certain $\emd$. This is given by:
\[
    \Delta_{\mathsf{EM}}(q_f) = \max_{K, K' \in \calX^*} \frac{\|q_f(K) - q_f(K')\|}{\emd(\tK, \tK')}.
\]
Naively, it is intractible to compute this sensitivity since there are exponentially many datasets of a given size. 
Additionally, this sensitivity might not always be bounded.  For instance, consider two points $x,x'$ that are close in $\calX$, but $f(x)$ is very far from $f(x')$. In this case,  we cannot bound $\Delta_{\mathsf{EM}}$, since the $K$ and $K'$ which put all their mass on $x$ and $x'$, respectively, will have $\frac{\|q_f(K) - q_f(K')\|}{\emd(\tK, \tK')} = \|f(x) - f(x')\|$. \add{We may exclude this case by assuming that $f$ is $\ell$-Lipschitz, meaning
\[
    \max_{x, x' \in \calX} \frac{\|f(x) - f(x')\|}{d_\calX(x,x')} \leq \ell.
\]
It turns out that Lipschitzness is precisely the property needed in order to bound $\|q_f(\tK) - q_f(\tK')\|$ while effectively accommodating the underlying coupling between $\tK, \tK'$. This is a novel aspect of our analysis that has not been explored by previous sensitivity analysis.
}
\begin{restatable}[]{thm}{sens}\label{thm:sens}
Let $q_f(K)$ be a linear query of the form in~\eqref{eq:embed-query}, where $f : \calX \rightarrow \R^d$ is $\ell$-Lipschitz. Then, we have $\Delta_{\mathsf{EM}}(q_f) \leq \ell$. 
\end{restatable}

\noindent
\textbf{Remarks.}
\add{The above result stands in contrast with traditional sensitivity arguments, which typically assume the weaker condition that $f$ is merely bounded by $\ell$. Although our Lipschitz assumption is stronger, it is satisfied by many queries of interest, such as 
\begin{itemize}
    \item \textit{Average.} If $\calX \subseteq \R^d$, then taking $f$ to be the identity function will cause $q_f$ to simply be the average of the elements in $K$.
    \item \textit{Density of a Kernel Smoothed Region.} If $R \subseteq \calX$ is a region of interest, then $f(x) = \textbf{1}[x \in R]$ will yield a $q_f$ computing the fraction of elements that lie in $R$. Using a smooth interpolation of $\textbf{1}[x \in R]$, such as kernel smoothers~\cite{hastie2009elements}, will give a smooth approximation to this value. 
    \item \textit{Similarity Queries}. If $\calX \subseteq \R^d$ and $c \in \R^d$ is a query vector, then letting $f(x) = \langle x, c \rangle$ will return the average similarity of each vector in $K$ with $c$, with similarity measured by the dot product. This is an instance of a \emph{linear embedding query} which we will study in detail in Section~\ref{sec:mean-est}.
\end{itemize}
}
Additionally, the aforementioned example illustrates that this sensitivity analysis is tight. This means that $d_\calX$, in addition to defining the privacy semantics, also influences the types of queries that can be answered with good utility.
\\\noindent
\add{ \textbf{Proof Sketch.} 
At a high level, consider moving $K$ onto $K'$ one point at a time. The Lipschitz assumption allows us to bound the change in $q_f(K)$ in terms of how far the point is moving, which translates to the $\emd$ distance. More formally, let $C$ be a minimum cost coupling between $\tK, \tK'$. We may write
\begin{align*}
    \|q_f(\tK) - q_f(\tK')\| = \left\|\E_{x \sim \tK}[f(x)] - \E_{y \sim \tK'}[f(y)]\right\|.
\end{align*}
We can view $x$ and $y$ as jointly generated by $C$, and the above expression becomes $\|\E_{x,y \sim C}[f(x) - f(y)]\|$. By the triangle inequality, we know this is at most
\[
    \|q_f(\tK) - q_f(\tK')\| \leq \E_{x,y \sim C}[\|f(x) - f(y)\|],
\]
and by the Lipschitz assumption, we may upper bound this by
\[
\ell \cdot \E_{x,y \sim C}[d_\calX(x, y)] = \ell \times \emd(\tK, \tK'),
\]
which completes the proof.  \qed
}

Using the upper bound on $\Delta_{\mathsf{EM}}(q_f)$, we follow a well-known approach for privately releasing a point with known sensitivity under a norm: sample a point $U$ uniformly from the ball $\{x \in \R^d : \|x\| = 1\}$, and release $q_f + \ell g U$, where $g \sim \Gamma(d, \frac{\omega}{\alpha})$ is  the Gamma distribution with shape $d$ and scale $\frac{\omega}{\alpha}$~\citep{hardt2009geometry}. Here, $\omega$ is a scale parameter that may be different in the central or local model, since the sensitivity of $f$ is less in the bounded central model. This mechanism, \Lin{}, is outlined in Algorithm~\ref{alg:lap}. Combining Theorem~\ref{thm:sens} with a standard privacy analysis, we can show that \Lin{} satisfies $(\alpha, 0)$-$\emd$ DP.
\begin{restatable}[]{lemma}{lqpriv}\label{lemma:lin}
    $\Lin$ (Algorithm~\ref{alg:lap}) with scale $\omega = \frac{1}{\alpha}$ satisfies $(\alpha, 0)$-unbounded local $\emd$-DP and with scale $\omega = \frac{1}{\alpha n}$ satisfies $(\alpha, 0)$-bounded central $\emd$-DP. 
\end{restatable}
\noindent
\textbf{Remarks.} When using the $1$-norm, \Lin{} becomes the multidimensional Laplace mechanism. We may instantiate \Lin{} with any noise mechanism that preserves $(\alpha, \delta)$-local $\|\cdot\|_p$-DP in the space $\R^d$. In Section~\ref{sec:mean-est}, we will instantiate \Lin{} using Gaussian noise of width $\frac{\omega \sqrt{1.25\ln(1/\delta)}}{\alpha}$~\cite{dwork2014algorithmic}, which will give the proper error dependence on $d$ under the $2$-norm.\\\\
\noindent
\textbf{Concrete Example.} In our location example, consider releasing the average distance of each point in $K$ from a particular city in the local model. This can be expressed with $f(x) = d_\calX(x, c)$, where $c$ is the city; by the triangle inequality this is $1$-Lipschitz. \Lin{} could then be applied to release $q_f(K_i)$ plus noise of expected magnitude $\frac{\ell}{\alpha} = 0.04$ per user; the total noise will be $\frac{0.04}{\sqrt{n}} = 1.26\times 10^{-4}$, corresponding to an error of just $1.6$ miles.

\RestyleAlgo{ruled}
\SetKwComment{Comment}{$\rhd$}{}
\begin{algorithm}
\caption{\Lin{}, an algorithm for releasing linear queries under bounded $\emd$-DP.}\label{alg:lap}
\KwData{$q_f$ -- A $d$-dimensional linear query; $\ell$ --  Upper bound of the Lipschitz constant of $f$; $K$ -- Input dataset; $\omega$ -- scale parameter}
\KwResult{An estimate of $q_f(K)$}
Sample $U$ uniformly from $\{x \in \R^d : \|x\| = 1\}$\;
Sample $g \in \R$ from $\Gamma(d, \omega)$\;
\KwRet{$\hat{q} = q_f(\tK) + \ell g U$}\;
\end{algorithm}
\subsection{Unordered Release of Item-wise Queries}\label{sec:hist-release}
We now consider the problem of directly releasing a private query applied to each item in $K$. This can provide a more fine-grained result than the aforementioned linear queries, which outputs the average over all the items. We release the query results as an unordered list to take advantage of the fact that subsequent computation (such as, aggregation) often does not depend on the ordering of the data~\cite{feldman2022hiding}.
Specifically, 
our second mechanism \Hist{} applies a mechanism $\calA$, which satisfies $(\alpha_0,0)$-$d_{\calX}$-DP, \emph{to each item} individually. We use $\calA$ as a black-box making \Hist{} completely general-purpose. For example, one could let $\calA$ be a private item-release mechanism (see Section~\ref{sec:rel-work} for some examples) and use \Hist{} to form a histogram of the dataset. $\calA$ could also be a classifer, and  \Hist{} can then release a simplified representation of the dataset. Once \Hist{} applies $\calA$ to each element in the dataset, it shuffles the results (to remove any ordering of the data) and outputs the shuffled list. This appears in Algorithm~\ref{alg:hist}, and a precursor appeared in~\cite{fernandes2019generalised}.

As \Hist{} does not hide the size of $K$, we show it satisfies bounded $\emd$-DP.
We use the following argument: for a neighboring dataset $K' = \{x_1', \ldots, x_m'\}$, by Lemma~\ref{lem:bvn} there exists a permutation $\pi : [m] \rightarrow [m]$ satisfying~Eq. \eqref{eq:bvn}.
Observe that we release the query responses in an unordered fashion by explicitly shuffling them. This allows us to pair up the element $x_i$ with $x_{\pi(i)}$ and analyze the privacy guarantee of releasing $\calA(x_1), \ldots, \calA(x_m)$ versus $\calA(x_{\pi(i)}'), \cdots, \calA(x_{\pi(m)}')$. Prior work does this with composition~\cite{fernandes2019generalised}.
 \RestyleAlgo{ruled}
\SetKwComment{Comment}{$\rhd$}{}
\begin{algorithm}
\caption{\Hist{}, a general mechanism for releasing a item-wise queries from $K$ as an unordered list under bounded $\emd$-DP}\label{alg:hist}
\KwData{Dataset $K \in \calX^m$, Mechanism $\calA : \calX \rightarrow \calY$ satisfying $(\alpha_0, 0)$-$d_{\calX}$ DP}
\KwResult{$L \in \calY^m$, unordered list (multiset) of item-wise queries from $K$}
$L = \emptyset$ \;
\For{$i=1, \ldots, m$}{
Add $\calA(x_i)$ to $L$\;
}
$\textrm{Shuffle}(L)$\;
\KwRet{$L$}
\end{algorithm}
However, composition is not the right tool for a tight privacy analysis. The reason is that composition assumes that each $\calA(x_i)$ is output sequentially, and in particular it is possible to identify which point came from $\calA(x_i)$ and which came from $\calA(x_{\pi(i)})$. In our case, we output an unordered list, and it is not possible to link which point came from an index $i$. Based on this observation, our key idea is to leverage privacy amplification by shuffling~\citep{feldman2022hiding} instead, which can yield a much smaller privacy parameter when the output is order invariant.

\add{
In particular, our core technical contribution is to analyze a general formulation of the privacy amplification by shuffling problem, where the vector $x_1, \ldots, x_m$ is changed to an arbitrary vector $x_1', \ldots, x_m'$. The key quantities we have control over are $\|v\|_1$ and $\|v\|_0$, where $v= (d_\calX(x_i, x_i'))_{i=1}^m$ (the $\emd$ distance allows us to bound $\|v\|_1$ and the maximum contribution from a user allows us to bound $\|v\|_0$); thus, our privacy bound depends on them. We believe our general result is of independent interest in the field of metric DP. Formally,
}
\begin{restatable}[]{thm}{shufflepriv}\label{thm:shuffle-metric-full}
    Suppose that $(\calX, d_{\calX})$ is a metric space such that $d_\calX(\cdot, \cdot) \leq 1$, and that $\calA$ is an ($\alpha_0,0)$ $d_\calX$-DP algorithm. Let $(x_1, \ldots, x_m)$ and $(x_1', \ldots, x_m')$ be two vectors, and we define  $v = (d_\calX(x_i, x_i'))_{i=1}^m$. Let $0 < \delta < 1$ be a constant, and suppose it holds that $\alpha_0 < \ln(\frac{m}{16 \ln (4m/\delta)})$. Then, for all sets $O$ of outputs, we have that
    \[
        \Pr[\mathsf{Shuffle}(\calA(x_1), \ldots, \calA(x_m)) \in O]  \leq e^{\alpha} \Pr[\mathsf{Shuffle}(\calA(x_1'), \ldots, \calA(x_m')) \in O] + \delta e^{\alpha},
    \]
    where 
    \begin{equation*}\label{eq:eps-amp-formula}
        \alpha \leq \|v\|_0\ln\left(1 + \frac{\exp(\sfrac{\alpha_0 \|v\|_1}{\|v\|_0})-1}{\exp{(\sfrac{\alpha_0 \|v\|_1}{\|v\|_0})}+1}\left(\frac{8 \sqrt{e^{\alpha_0} \ln (4\|v\|_0 / \delta)}}{\sqrt{m}} + \frac{8 e^{\alpha_0}}{m} \right)\right).
    \end{equation*}
\end{restatable}
\noindent
\textbf{Remarks.} 
    In particular, if $\alpha_0 \leq \frac{\|v\|_0}{\|v\|_1}$, the above bound is $\approx \frac{\alpha_0 \|v\|_1}{\sqrt{m}}$, which grows with just $\sqrt{m}$ (as $\|v\|_1 \leq m$).
    \add{The standard shuffling result only assumes that $\calA$ satisfies $\alpha$-\emph{local} DP, and that just $x_1$ is changed to $x_1'$ (since each user owns a single item). Theorem~\ref{thm:shuffle-metric-full} can be specialized to recover the state-of-the-art result for this special case~\cite{feldman2022hiding}, but it is  significantly more general in its current form.
    }
\\\add{
\noindent\textbf{Proof Sketch.}
Using group privacy, we can analyze the privacy guarantee between $(x_1, \ldots, x_m)$ and $(x_1', \ldots, x_m')$, where up to $\|v\|_0$ points change, instead of just changing one point as a time. We then analyze a change of one point by generalizing and simplifying the state-of-the-art technique in~\cite{feldman2022hiding}. We show that the resulting privacy parameter for changing the point $x_i$ to $x_i'$ is $g(w_i)$ where 
\[
g(w_i) = \ln\left(1 + \frac{e^{\alpha_0 w_i}-1}{e^{\alpha_0 w_i}+1}\left(\frac{8 \sqrt{e^{\alpha_0} \ln (4 / \delta)}}{\sqrt{m}} + \frac{8 e^{\alpha_0}}{m} \right)\right).
\]
By group privacy $\|v\|_0$ times, the overall privacy parameter is $\sum_{i=1}^{\|v\|_0} g(w_i)$. The final step is proving that $g$ is concave so the worst-case amplification is simply $\|v\|_0 g(\frac{\|v\|_1}{\|v\|_0})$.  
\qed 
}

\noindent
\textbf{Comparison with Composition.} Analyzing Theorem~\ref{thm:shuffle-metric-full} using the state-of-the-art composition results~\citep{kairouz2015composition} and $\alpha_0 \leq 1$ gives us
\[
\alpha \leq O\left(\alpha_0 \|v\|_2 \sqrt{\ln\tfrac{1}{\delta}}\right).
\]
However, we cannot form a satisfying bound on the $2$-norm of $v$---it is only possible to say $\|v\|_2 \leq \|v\|_1$ which is tight when e.g. $\|v\|_0$ = 1. The bound is thus missing the factor of $\frac{1}{\sqrt{m}}$---composition here does not leverage the fact that all $m$ items are released in a random order.

Combining~\eqref{eq:bvn} and Theorem~\ref{thm:shuffle-metric-full}, we obtain an improved privacy guarantee for \Hist{}. The guarantee can be stated in both the bounded local and central models. In the local model, recall that each user is applying \Hist{} to their data. In the central model, the aggregator applies \Hist{} to the entire dataset, and releases the frequencies of $mn$ itemwise queries.

\begin{restatable}[]{thm}{epsamp}\label{thm:eps-amp}
    For any $\delta \in (0,1)$, \Hist{} shown in Algorithm~\ref{alg:hist} satisfies bounded local $(\alpha, \delta')$-$\emd$ DP, where \[
    \textstyle{\alpha = \sup_{w \in [0,1]} \frac{h(m; m, mw)}{w} \ \ \ \ \ \text{ and } \delta' = \delta e^{h(m; m, m)},}
    \]
    and
    \[
    h(m; x_0, x_1) = x_0\ln\left(1 + \frac{\exp(\sfrac{\alpha_0 x_1}{x_0})-1}{\exp{(\sfrac{\alpha_0 x_1}{x_0})}+1}\left(\frac{8 \sqrt{e^{\alpha_0} \ln (4x_0 / \delta)}}{\sqrt{m}} + \frac{8 e^{\alpha_0}}{m} \right)\right).
    \]
    Similarly, \Hist{} satisfies bounded central $(\alpha, \delta')$-$\emd$ DP, where 
    \[ \textstyle{\alpha = \sup_{w \in [0, 1]} \frac{h(mn; m, mw)}{w} \ \ \ \ \ \text{ and } \delta' = \delta e^{h(mn; m, m)}.}
    \]
\end{restatable} 
\noindent
\textbf{Remarks.} Theorem~\ref{thm:eps-amp} gives the tightest possible privacy parameters, but we may also give an asymptotic formula as follows. For desired privacy parameters $(\alpha, \delta)$, one should set
\begin{gather}\label{eq:eps-amp}
    \alpha_0 = \begin{cases} \frac{\alpha}{32\sqrt {m \ln(4m e^{\alpha}/ \delta)}} & \text{ if $\alpha \leq 32 \sqrt{m \ln (4m e^{\alpha}/ \delta)}$} \\ 2\ln \left( \frac{\alpha}{16\sqrt{m  \ln(4m e^{\alpha} / \delta)}} \right) & 32 \sqrt{m \ln (4m e^{\alpha}/ \delta)} < \alpha < m \end{cases} \end{gather}
    and 
    \begin{gather}\label{eq:eps-amp-central}
    \alpha_0 = \begin{cases} \frac{\alpha \sqrt{n}}{32\sqrt {m \ln(4m e^{\alpha}/ \delta)}} & \text{ if $\alpha \sqrt{n} \leq 32 \sqrt{m \ln (4m e^{\alpha}/ \delta)}$} \\ 2\ln \left( \frac{\alpha \sqrt{n}}{16\sqrt{m  \ln(4m e^{\alpha} / \delta)}} \right) & 32 \sqrt{m \ln (4m e^{\alpha}/ \delta)} < \alpha \sqrt{n} < m \sqrt{n} \end{cases}
\end{gather}
in order to achieve $\emd$-DP in the bounded local and central model, respectively.
  Assuming $\alpha \leq O(\ln (m))$, this means that the privacy parameter will be roughly $\frac{\alpha}{\sqrt{m}}$ (resp. $\ln(\frac{\alpha \sqrt{n}}{\sqrt{m}})$) for releasing the $m$ samples; this is asymptotically better than the analysis with composition which would require setting $\alpha_0 = \frac{\alpha}{m}$ (resp. $\frac{\alpha}{m})$). Even with higher $\alpha = m^c$ for $c < 1$, the budget is still $\frac{\alpha}{\sqrt{m^{1+c}}}$ (resp. $\ln(\frac{\alpha \sqrt{n}}{\sqrt{m^{1+c}}}$)), which are both significant asymptotic improvements. 
\\\noindent
\textbf{Concrete Example.} Our improved analysis makes the most significant improvements in the central model. Here, we would have to apply \Hist{} with $\alpha_0 = \frac{\alpha}{m} = 0.025$ for each of the $m=10^3$ location data points per user. Using the guarantee of Theorem~\ref{thm:eps-amp}, it is possible to set $\alpha_0 \approx 3.0$ -- a several orders of magnitude improvement. 
\\\noindent \add{\textbf{Open Questions.} The primary open question in the design of our mechanism is whether one can obtain a tighter privacy analysis of Theorem~\ref{thm:eps-amp} that does not rely on group privacy, but rather analyzes the amplification potential of all points at once. Another interesting direction is to prove Theorem~\ref{thm:shuffle-metric-full} where the LDP mechanism satisfies approximate metric DP. We believe similar techniques should work for this setting, and leave it for future work.}

\section{Generalization to Unbounded DP}\label{sec:reduct}

The mechanisms presented so far face two challenges when applied to the unbounded data setting. First, a direct privacy analysis of the unbounded data setting is difficult since we cannot leverage Lemma~\ref{lem:bvn}, which significantly simplifies the analysis (for the bounded data setting). 
Second, and more importantly, the unbounded central model offers no utility improvement over the local model. In the worst-case scenario, a single user may contribute nearly all the data in the dataset, effectively reducing any algorithm to satisfying only local $\emd$-DP. This issue has been noted in previous work in user-level DP~\cite{liu2023algorithms}.

In this section, we tackle these challenges by showing a \textit{blackbox} reduction from unbounded $\emd$-DP to bounded $\emd$-DP. Our reduction works in both the local and central models. The key idea of the reduction is to \emph{smoothly} project a dataset $K$ of any size to a dataset $L$ of a given fixed size,  such that the $\emd$ distance between any two input datasets and the $\emd$ distance between their projections are roughly the same. Then, it is easy to show that applying any bounded $\emd$-DP algorithm to the smooth projections is sufficient to guarantee unbounded $\emd$-DP for the entire scheme. Full proofs for this section appear in Appendix~\ref{app:reduct-proofs}.

Our proposed projection mechanism is smooth in a near-multiplicative sense, albeit with a small additive penalty when the $\emd$ between the two datasets is small. We account for this subtlety by slightly modifying the privacy semantics of $\emd$-DP in the unbounded setting to not grow arbitrarily strong as $\emd(\tK, \tK') \rightarrow 0$. Instead, we introduce a distance threshold $r$ such that all $\emd(\tK, \tK') \leq r$ enjoys a \textit{uniform} privacy guarantee of $\epsilon$. This refined privacy definition, termed discrete $\emd$-DP, is formalized (in the local model) as:

\begin{defn}\label{def:discrete-emd-dp-local}
    [Discrete Local $\emd$-DP] Let $\calM$ be a mechanism which acts on a dataset $K$. We say $\calM$ satisfies $(\epsilon, \delta, r)$-discrete local $\emd$-DP if, for any two datasets $K, K' \in \calX^*$ such that $\emd(\tK, \tK') \leq r$ and any set of outputs $O$,
    \begin{equation*}
        \Pr[\calM(K) \in O] \leq e^{\epsilon} \Pr[\calM(K') \in O] + \delta.
    \end{equation*}\label{def:discrete}
\end{defn}
Like in standard DP, the above definition uses the parameter $\epsilon$ because it is a \emph{unitless} privacy parameter---the unit of the metric is expressed in the parameter $r$.  
\begin{fact} For any $K,K'$ such that $d=\emd(\tilde{K},\tilde{K}')$, $\calM$ satisfies 
 \begin{equation*}
        \Pr[\calM(K) \in O] \leq e^{\epsilon \lceil \frac{d}{r} \rceil} \Pr[\calM(K') \in O] + \delta\left\lceil \frac{d}{r} \right\rceil \exp(\lceil \tfrac{d}{r} \rceil).
    \end{equation*}\label{ref:group}
\end{fact}
Fact~\ref{ref:group} is implied from Definition \ref{def:discrete} followed by a direct application of group privacy~\cite{vadhan2017complexity}. This guarantee can be interpreted as providing $\emd$-DP at the granularity of units of $\emd$ distance $r$. Note that for all $d \geq r$, we have $\epsilon \lceil \frac{d}{r} \rceil \leq \frac{2\epsilon}{r} d$. Thus, $(\epsilon, \delta, r)$-discrete local $\emd$ DP is roughly equivalent to $(\frac{2\epsilon}{r}, \delta)$-unbounded local $\emd$-DP, except if $\emd(\tK, \tK') \leq r$. In this case, the privacy parameter will not go below $\epsilon$. This adjustment does not significantly alter the overall privacy semantics of $\emd$-DP; one may simply set $\alpha$ as described in Section~\ref{sec:mechanisms}.

In the central model, we make a similar definition:
\begin{defn}\label{def:discrete-emd-dp-central}
    [Discrete Central $\emd$-DP] Let $K_G = K_1 \cup \cdots \cup K_n$ denote a global dataset from $n$ users (of any size). We say $K_G \sim_r K'_G$ if $K'_G$ can be obtained from $K_G$ by changing $K_i$ to $K_i'$ for just one user $i$, such that $\emd(\tK_i, \tK'_i) \leq r$. We say a mechanism $\calM(K_G)$ satisfies $(\epsilon, \delta, r)$-discrete central $\emd$-DP if, for all $K_G, K_G'$ such that $K_G \sim_r K_G'$ and any set of outputs $O$, we have 
    \begin{equation*}
        \Pr[\calM(K_G) \in O] \leq e^{\epsilon} \Pr[\calM(K'_G) \in O] + \delta.
    \end{equation*}
\end{defn}
As before, $(\epsilon, \delta, r)$-discrete central $\emd$-DP is roughly equivalent to $(\frac{2\epsilon}{r}, \delta)$-bounded central $\emd$-DP when all user datasets have size $m$. We will see that Definition~\ref{def:discrete-emd-dp-central} is the appropriate generalization to unbounded user datasets under our projection mechanism which is described below.

\add{Because our projection mechanism must preserve the $\emd$ between $\tK, \tK'$, it is not acceptable to select an arbitrary set of points from $\tK, \tK'$, as this could dramatically inflate the $\emd$ distance. However, couplings are intimately tied with \emph{sampling} -- observe that two random samples from $\tK, \tK'$ can actually be coupled so that their expected distance is $\emd(\tK, \tK')$. If this process is repeated multiple times, then the cost of coupling the samples will converge to $\emd(\tK, \tK')$. Thus, sampling with replacement will result in a projection that does not increase the $\emd$ distance by too much.
}

\begin{algorithm}
\caption{\ReductCentral{}, a reduction from unbounded $\emd$-DP to bounded $\emd$-DP.}\label{alg:red-central}

\KwData{$K_G$ - Global datasets of $n$ users; $\calA$ - A mechanism satisfying bounded $\emd$-DP; $s$ - Number of samples.}
$L = \emptyset$\;
\For{$i=1$ to $n$}{
Add $s$ uniform samples with replacement from $K_i$ to $L$
}
$O = \calA(L)$\;
\KwRet{$O$}

\end{algorithm}
\begin{restatable}[]{lemma}{smoothproj}\label{lem:coupling-dist}
    Let $\tK, \tK' \in \Delta^\calX$ be probability distributions, and let $C^*$ be the minimum cost coupling between $\tK, \tK'$. Let $\{(x_i, y_i)\}_{i=1}^s$ be $s$ i.i.d. samples from $C^*$, $L = (x_1, \ldots, x_s)$ and $L' = (y_1, \ldots, y_s)$. Then,
\[
    \Pr[\emd(\tL, \tL') \geq (1+\sqrt{2}) \emd(\tK, \tK') + \tfrac{3}{s} \ln (\tfrac{1}{\delta})] \leq \delta.
\]
\end{restatable}
\add{
\noindent\textbf{Remarks.} 
\noindent The multiplicative factor of $1+\sqrt{2}$ shows that the projection is smooth when $\emd(\tK, \tK')$ dominates the additive factor of $\frac{3}{s} \ln(\frac{1}{\delta})$. This is achieved when the number of samples $s$ is much larger than the inverse of $\emd(\tK, \tK')$. If $s$ is not sufficiently large, then not enough samples are being taken to ensure convergence.
}
\\
\add{
\noindent\textbf{Proof Sketch.}
Let $C$ be a minimum-cost coupling between $\tK, \tK'$.
To show that $\tL, \tL'$ are not far from each other, we can view $\tL'$ as being generated from $\tL$, where for each point $x \in L$, a point $y \sim C_{x}(\cdot)$ is added to $L'$. This view of $\tL, \tL'$ shows there is a transportation plan from $L = \{x_1, \ldots, x_s\}$ and $L' = \{y_1, \ldots, y_s\}$ of expected cost $\E_{x \sim C_1, y \sim C_x} d_{\calX}(x, y) = \emd(\tK_i, \tK_i')$. Using Bernstein's inequality, we can show with probability at least $1-\delta$, $\emd(\tL, \tL')$ is upper bounded by $2 \emd(\tK_i, \tK_i') + \frac{6}{s} \log(\frac{1}{\delta})$.   \qed
}

\add{
Our full reduction to bounded $\emd$-DP first projects $K$ onto a dataset of size $m$ by taking samples with replacement. Next, it applies a \textit{blackbox} bounded $\emd$-DP mechanism, $\calA$, to the projected dataset $L$. By blackbox application we mean that $\calA$ can be any arbitrary mechanism as long as it satisfies bounded $\emd$-DP. We call this mechanism \ReductCentral{}, and it is illustrated in the central model in Algorithm~\ref{alg:red-central} (in the local model, each user samples from their own $K_i$, so we would simply have $n=1$).
Using the projection guarantee of Lemma~\ref{lem:coupling-dist}, \ReductCentral{} enjoys the following privacy guarantee:
}
\begin{restatable}[]{thm}{reductpriv}\label{thm:emd-dp-central}
    Let $\epsilon > 0$ and $\delta, r \in [0,1]$ be arbitrary constants. Suppose $\calA$ is a mechanism which satisfies $(\alpha, \delta)$-bounded local $\emd$-DP (Definition~\ref{def:emd-dp}), where
    \[
        \alpha = \tfrac{\epsilon}{(1 + \sqrt{2}) r + \tfrac{3}{s} \ln(\tfrac{1}{\delta})}.
    \]
    Then, \ReductCentral{} satisfies $(\epsilon, 2\delta, r)$-discrete local $\emd$-DP. Similarly, if $\calA$ is $(\alpha, \delta)$-bounded central $\emd$-DP (Definition \ref{def:emd-dp}), then \ReductCentral{} is $(\epsilon, 2\delta, r)$-discrete central $\emd$-DP.
\end{restatable}
\noindent
\textbf{Remarks.}
If the number of samples $s$ is at least $\frac{\ln(1/\delta)}{r}$, then Theorem~\ref{thm:emd-dp-central} shows there is only a small multiplicative cost to considering just bounded $\emd$-DP (in the respective local or central model). In this case, $\calA$ will need to roughly satisfy $(\frac{\epsilon}{r}, \delta)$-bounded $\emd$-DP, and this is roughly the same as the resulting $(\epsilon, \delta, r)$-discrete $\emd$ DP algorithm. There is no privacy disadvantage to taking a large number of samples, and the utility may also increase due to more information about the dataset being captured (recall that the projection does not providing privacy; it is being provided by $\calA$). Thus, the number of samples may be set to be large with computational costs being the only constraint.\\
\noindent
\add{
\textbf{Proof Sketch.}
The proof of privacy is almost immediate from Lemma~\ref{lem:coupling-dist}. The variables $\tL, \tL'$ in two executions of \ReductCentral{} are random variables, but with probability $\delta$, the $\emd$ distance between them is $c=(1+\sqrt{2})r + \frac{3}{s}\ln(\frac{1}{\delta})$. By the convexity of differential privacy, we can analyze the privacy parameter of every fixed choice of $\tL, \tL'$. With probability $1-\delta$, the privacy parameter will be $\frac{\alpha}{c}$, by the privacy guarantee of $\calA$.
\qed}\\\noindent
\ReductCentral{} can be used to bound the contribution of each user in the central setting, allowing us to apply the simpler Definition \ref{def:emd-dp-central}. In addition, it can be used to adapt \Hist{} to the unbounded data setting. One caveat is that utility may not be preserved if the number of user samples is too small or, in the central setting, if the users data distributions are heterogeneous. In particular, if users have varying numbers of samples, each from different distributions, applying \ReductCentral{} equalizes the frequency of all user data. Nonetheless, it is often reasonable to assume the users have homogeneous data distributions~\cite{liu2020learning, acharya2023discrete}. 
\\\noindent
\add{\textbf{Open Questions.} Many sampling procedures are likely to be compatible with $\emd$. This leads to the question of whether different procedures, such as sampling without replacement, are also smooth projections.}

\section{Applications of Proposed Mechanisms}\label{sec:utility}

In this section, we compare the utilities of \Lin{} and \Hist{} to existing mechanisms satisfying user-level DP. 
For simplicity, we assume the bounded data setting.
Full proofs for this section appear in Appendix~\ref{app:uti-proofs}.

\textbf{Notations.}
We define the following quantities of a real matrix $M \in \R^{d \times k}$. First, the $(p,q)$ operator norm of $M$, denoted by $\|M\|_{p \rightarrow q}$, is given by $\|M\|_{p \rightarrow q} = \sup_{x \in \R^k, \|x\|_p \leq 1} \|Mx\|_q$. We can show that $\| M \|_{1 \rightarrow 2}$ is equal to the maximum $\ell_2$ norm of a column of $M$. Furthermore, $\|M\|_{2 \rightarrow 2}$, more commonly written as $\|M\|_2$, is the spectral norm and is equal to the maximum singular value of $M$. Matrix norms satisfy the important submultiplicative property, which states that $\|MN\|_{p \rightarrow r} \leq \|M\|_{q \rightarrow r} \|N\|_{p \rightarrow q}$ for any matrices $M,N$ and $p,q,r \geq 1$. 
Next, let $I_d$ denote the $d \times d$ identity matrix, and again suppose that $M \in R^{d \times k}$ with $d \leq k$. If $M$ has full row rank, then there exists a matrix $N \in \R^{k \times d}$ such that $MN = I_d$. We call such a matrix $N$ a \emph{right inverse} of $M$. Finally, for $M \in \R^{s_1 \times t_1}$ and $N \in \R^{s_2 \times t_2}$, let $M \otimes N \in \R^{s_1s_2 \times t_1t_2}$ denote the Kronecker product of two real matrices, whose entry in $((i_1,i_2), (j_1,j_2))$ is $M_{i_1j_1}N_{i_2j_2}$. 

\subsection{Linear Embedding Queries}\label{sec:mean-est}

Many applications of metric DP assume there is an embedding function $\phi : \calX \rightarrow \R^t$, which maps an item to its semantic representation in $\R^t$ (each of the examples in Section~\ref{sec:intro} have an embedding representation). The metric $d_\calX$ is then the distance between $\phi(x)$ and $\phi(x')$; in this section, we consider the $l_2$ distance.

Since $\phi(x)$ also communicates information about the item $x$, we define \emph{linear embedding queries} as linear queries applied to an item's embedding $\phi(x)$. Formally,
\[
    q_{f \circ \phi}(K) = \E_{x \sim \tK} [f \circ \phi (x)],
\]
where $f(y) = Fy$ for a matrix $F \in \R^{d \times t}$ (meaning that $f$ is a \emph{linear} function). Assume each row $F_i$ of $F$ is normalized so that $\|F_i\|_2 \leq 1$. Each coordinate of $f \circ \phi $ is equal to $\E_{x \sim \tK} [\langle F_i, \phi(x)\rangle]$. \add{Thus, we may view each coordinate of a linear embedding query as a similarity query in the embedding space with query point $F_i$.}
For ease of discussion, our analysis will assume that $d < |\calX|$ and $d \ll n$, which is often a setting of practical interest. Note that we may write $q_{f \circ \phi}$ as $F \Phi \tK$, where  $\Phi \in \R^{t \times \calX}$ is the collection of embedding vectors in $\calX$. 

\subsubsection{Local Model} Existing user-level DP mechanisms ask user at index $i$ to privately release the query $\hat{q}_i = q_{f \circ \phi}(\tK_i)$. The aggregator computes the average $\hat{q} = \frac{1}{n}\sum_{i=1}^n \hat{q}_i$. The best-known solutions for the low-$d$ setting are based on additive noise mechanisms~\cite{duchi2013local, bassily2019linear}, or more sophisticated matrix factorization mechanisms~\cite{edmonds2020power}. We focus on the former algorithms; these results also carry over to matrix factorization algorithms. In low-dimension, the best additive noise mechanisms achieve the following error:

\begin{lemma}\label{lem:lin-user} \cite[Proposition 3]{duchi2013local}
    There exists an $(\epsilon, 0)$-bounded user-level DP mechanism in the local model algorithm which produces an estimate $\hat{q}$ such that, for all $\tK$, 
    \[
        \E[\|\hat{q} - q_{f \circ \phi}(\tK)\|_2] \leq O\left(\|F \Phi\|_{1\rightarrow 2} \tfrac{\sqrt{d}}{\epsilon \sqrt{n}} \right).
    \]
\end{lemma}
To interpret the term $\|F \Phi\|_{1 \rightarrow 2}$, we can use the inequality $\|F \Phi\|_{1 \rightarrow 2} \leq \|F\|_2 \|\Phi\|_{1 \rightarrow 2}$, which is tight for certain choices of $F$ and $\Phi$. By assumption, we know $\|F\|_2 \leq \sqrt{d}$ and $\|\Phi \|_{1\rightarrow 2} \leq 1$, both of which can also be tight. The bound is thus $O(\frac{d}{\epsilon \sqrt{n}})$.

On the other hand, for $\emd$-DP, by Theorem~\ref{thm:sens}, we know that $\Delta_{\mathsf{EM}}(q_{f \circ \phi})$ is at most the Lipschitz constant of $f \circ \phi$ given by:
\[
    \max_{x, x' \in \calX} \tfrac{\|F (\phi(x)) - F(\phi(x')) \|}{\|\phi(x)-\phi(x')\|} \leq \max_{x, x' \in \calX} \tfrac{\|F(\phi(x)) - \phi(x'))\|}{\|\phi(x)-\phi(x')\|} \leq \|F\|_{2}.
\]
Hence, each user can apply \Lin{} with the Gaussian mechanism with $\ell = \|F\|_2$, which gives the following utility guarantee:
\begin{restatable}[]{lemma}{lqacc}\label{lem:lq-local}
    There exists an $(\alpha, \delta)$-bounded  $\emd$-DP algorithm in the local model which produces an estimate $\hat{q}$ such that, for all $\tK$, 
    \[
        \E[\|\hat{q} - q_{f \circ \phi}(\tK)\|_2] \leq O(\|F\|_2\tfrac{\sqrt{d \ln (1/\delta))}}{\alpha \sqrt{n}}).
    \]
\end{restatable}\noindent
\textbf{Remarks.}
We use the Gaussian mechanism because it performs better under the $\ell_2$ error than the pure $(\alpha, 0)$-bounded local $\emd$ DP illustrated in Algorithm~\ref{alg:lap}. However, this forces us to use $\delta > 0$. 
Compared to Lemma~\ref{lem:lin-user}, and assuming that $\|F\Phi\|_{1 \rightarrow 2} \approx \|F\|_2$, the above bound differs by a factor of $\frac{\epsilon}{\alpha}$ (and small $\ln \frac{1}{\delta}$ terms)---when $\alpha = \epsilon$, we know that $\emd$-DP provides better privacy. When $\epsilon \ll \alpha$, that \Hist{} offers lower error than Lemma~\ref{lem:lin-user}.
\subsubsection{Central Model}
In the central model, linear query release has been extensively studied, and optimal algorithms under item-level DP are known~\citep{hardt2009geometry, bhaskara2012unconditional, nikolov2013geometry, edmonds2020power}. These algorithms can be easily adapted to user-level pure DP, providing the following guarantee\footnote{\add{For simplicity, we do not state or compare to the exact instance-optimal upper bounds known for linear queries. Instead, the upper bound in Lemma~\ref{lem:lin-user-cent} is the optimal one for random linear queries~\cite{hardt2009geometry}. Like with local DP, it is possible to carry over the sensitivity reduction of $\emd$-DP to other techniques such as matrix factorization.}}:
\begin{lemma}\label{lem:lin-user-cent} (From Theorem 1.3 in~\cite{hardt2009geometry})
    There exists an $(\epsilon, 0)$-bounded user-level DP algorithm in the central model which produces an estimate $\hat{q}$ such that, for all $\tK$,
    \[
        \E[\|\hat{q} - q_{f \circ \phi}(\tK)\|_2] \leq O\left(\|F \Phi\|_{1 \rightarrow 2} \tfrac{\sqrt{d}}{\epsilon n} \ln\left(\tfrac{k}{d}\right) \right).
    \]
\end{lemma}%
To provide $(\alpha, \delta)$-bounded $\emd$-DP in the central model, we can use \Lin{} with the Gaussian mechanism with scale $\omega = \frac{1}{n\alpha }$. Following the same approach as in Lemma~\ref{lem:lq-local}, this results in $O(\|F\|_2 \frac{\sqrt{d \ln \frac{1}{\delta}}}{\alpha n})$ error. Again, this is worse than Lemma~\ref{lem:lin-user-cent} by a factor of $\frac{\epsilon}{\alpha}$, and similar observations apply.


\subsection{Frequency Estimation}\label{sec:freq-est}
Here, we evaluate the error of \Hist{} for private frequency estimation, where the goal is to obtain a private estimate $\tH$ of the (normalized) histogram $\tK$. This problem has been extensively studied in privacy~\cite{hay2009boosting, xu2013differentially,suresh2019differentially, kairouz2016discrete, acharya2018communication, chen2020breaking, acharya2023discrete}; the high-level goal is to minimize the $\ell_p$ distance between $\tH$ and $\tK$, where $p$ is typically set to $1, 2$ or $\infty$. However when the data domain is a general metric space $\calX$, not all $\ell_p$ perturbations to $\tK$ are the same. Therefore, we we will measure the similarity between $\tK, \tH$ via $\emd(\tK, \tH)$, as we do in our privacy definition.

To simplify the analysis while still demonstrating the effectiveness of our mechanisms, we fix $\calX$ to be the following ``clustered'' metric space. Let $\calX = \calB \times \calC$, where $\calB = \{b_1, \ldots, b_s\}$, $\calC = \{c_1, \ldots, c_t\}$ and $s\cdot t = k$. For some $r < \frac{1}{2}$, the distance is given by the following:
\[
d_{\calB \times \calC} ((b,c), (b',c')) = \begin{cases} 
    0 & \text{if $b = b'$ and $ c = c'$} \\
    r & \text{if $b = b'$} \\
    1 & \text{otherwise}.
\end{cases}
\]
We can think of this metric space as a collection of $s$ clusters consisting of the $t$ items $\{(b, c_1), \ldots, (b, c_t)\}$ for each $b \in \calB$. Points in a cluster are more related, being at distance $r$ apart, than items in two different clusters, which are distance $1$ apart. We will assume that privacy level $\epsilon$ is needed between two items in the same cluster, so we will set $\alpha = \frac{\epsilon}{r}$.

\subsubsection{Algorithms in the Local Model}

At a high level, in the local model each user applies a private mechanism $\calA : \calX \rightarrow \calY$ (with $\calY$ discrete and $|\calY| \geq |\calX|$) to each sample and releases it. The central server forms an aggregate vector $v \in \R^\calY$. Let $A \in \R^{\calX \times \calY}$ denote the transition probability matrix of $\calA$; we have by linearity of expectation that $\E[v] = \tK A$. Local-DP imposes constraints on each pair of rows of $A$. Assuming that $A$ has a right inverse $B$, the central server returns the estimate $\tH = v B$, which is unbiased. Many mechanisms in local DP, including the ones with optimal error, can be expressed with this linear-algebraic view~\citep{erlingsson2014rappor, kairouz2016discrete, acharya2018communication, feldman2022private}. We summarize this in Algorithm~\ref{alg:hist-local}.

\begin{algorithm}
\caption{\FreqEstLocal{}, a general framework for histogram estimation under local DP}\label{alg:hist-local}
\KwData{$K$, a family of datasets from $n$ users each with size $m$; $\calA$, a mechanism from $\calX$ to $\calY$; $B \in \R^{\calY \times \calX}$, a right inverse of $\calA$.}
\For{\textbf{each} user $i$ from $1$ to $n$}{
    $L_i = \emptyset$\;
    \For{$l_j \in K_i$}{
        $r_j = \calA(l_j)$\;
        Add $r_j$ to $L_i$\;
    }
    Release $\tL_i$\;
}
$v = \frac{1}{n} \sum_{i=1}^n \tL_i$\;
$\tH = v B$\;
\Return{$\tH$}
\end{algorithm}

A state-of-the-art technique we will utilize is Hadamard response, which obtains order-optimal error~\cite{acharya2018communication, chen2020breaking}. This mechanism is based off of using Hadamard matrices to robustly encode $\calX$---specifically, the matrix $A$ is given by $q_1\textbf{1} + q_2H$, where $H$ is a Hadamard matrix and $q_1, q_2$ are constants chosen so that $A$ is normalized and that each element is proportional to either $e^{\epsilon}$ or $1$ to satisfy LDP. This mechanism has the following utility:

\begin{lemma}\label{lem:local-hist-utility-baseline} (From Theorem 3.1 in~\cite{chen2020breaking})
    There exists a mechanism $\calA$ such that $\FreqEstLocal$ satisfies $(\epsilon, \delta)$-bounded user-level DP and returns an estimator $\tH$ such that
    \[
        \max_{K} \E[\emd(\tK, \tH)] \leq O \left( \sqrt{\tfrac k {mn}} + \sqrt{ \tfrac{k^2 \ln (m/\delta))}{n\epsilon^2}}\right).
    \]
\end{lemma}
\textbf{Remarks.}
In order to adapt the Hadamard response to the user-level setting, we suppose each user applies $\calA$ to each sample with privacy budget $\frac{\epsilon}{\sqrt{m \ln (m/\delta)}}$, and $(\epsilon, \delta)$-user level DP follows from composition~\cite{kairouz2015composition}.
 The term $\sqrt{\frac{k}{mn}}$ is the sampling error which does not depend on $\epsilon$, and the second $\sqrt{\frac{k^2 \ln (m/\delta)}{n \epsilon^2}}$ term is the cost of privacy. The cost of privacy usually dominates, and furthermore its dependence on $m$ is not significant. This is because $m$ reduces both the effect of each sample on the final estimate, and the privacy budget per sample, countervailing itself.

\add{With $\emd$-DP under our chosen metric, we can use a transition probability matrix $A$ that is \emph{less noisy}. This comes from the fact that only items in the same cluster need a strong privacy guarantee, whereas traditional user-level DP would require all points to have this level of privacy guarantee.} To provide formal guarantees, we first derive an error bound on \FreqEstLocal{} in terms of $A$ (specifically its right inverse), which we will then optimize later.
\begin{restatable}[]{thm}{freqacc}\label{thm:general-hist-utility} 
    For the metric space $\calX = \calB \times \calC$ and any mechanism $\calA$ satisfying ($\alpha_0,0)$ $ d_\calX$-DP where $\alpha_0 = O(\frac{\alpha}{ \sqrt{m \ln(me^\alpha / \delta)}})$ ($\alpha_0$ is specifically defined in Theorem~\ref{thm:eps-amp}), \FreqEstLocal{} is $(\alpha, \delta)$-bounded $\emd$-DP in the local model and returns an estimator $\tH$ such that
    \begin{equation}\label{eq:emd-bound}
        \max_{K}\E[\emd(\tH, \tK)] \leq r \sqrt{\frac{st(\|B^T\|_{1 \rightarrow 2}^2 - 1)}{mn}} + \sqrt{\frac{s(\|P^TB^T\|_{1 \rightarrow 2}^2 - 1)}{mn}},
    \end{equation}
    where $B$ is a right inverse of $\calA$, $P = I_{\calB} \otimes 1^+_{\calC}$, and $1^+_{\calC}$ is a column vector of $1$s indexed by $\calC$. 
\end{restatable}
\noindent
\textbf{Remarks.}
The first term in the RHS of Eq.~\eqref{eq:emd-bound} is the cost of equalizing the mass between clusters, and the second term is the cost of equalizing the mass across clusters (since the matrix $P$ essentially projects $\calA$ to act between clusters). For small $r$, the first term approaches $0$, and the latter term may also approach $0$ because $\calA$ will not often map a point outside its cluster under $d_\calX$-DP (and thus, $\|P^TB^T\|^2_{1 \rightarrow 2} - 1 \rightarrow 0$).\\
\noindent
\add{\textbf{Proof Sketch.} Observe that we may upper bound $\emd(\tH, \tK)$ with any transportation plan between $\tH, \tK$. We will use the following one: first map the probability masses in each cluster so that they match, putting extra mass in an arbitrary point. This incurs at most $r \|\tH - \tK\|_1$ cost, since the intra-cluster distance is at most $r$. Next, equalizing the mass between clusters, which incurs at most $\|P(\tH - \tK)\|_1$ cost, where $P$ is the given matrix which can be viewed as the linear operator that adds the mass within each cluster together. Both of the error terms can then be bounded by viewing $\tH - \tK$ as the sum of $mn$ independent variables drawn from a Dirichlet distribution with mean $0$, and applying a variance analysis.\qed\\\\}
\add{\noindent Now, the task is to pick a mechanism $\calA$ satisfying $\emd$-DP, which minimizes the error in~\eqref{eq:emd-bound}. The constraint of $\emd$-DP is quite different from standard DP, and permits novel mechanism design. We use a natural generalization of $k$-randomized response~\cite{kairouz2016discrete}, adapted to $d_\calX$-DP.} Specifically, $\krr_{\alpha_0}$ has probabilities given by, for each $(b,c) \in \calX$,
\begin{gather*}
    \Pr[\krr(b,c) = (b,c))] \propto e^{\alpha_0}, \\ \Pr[\krr(b,c)) = (b, c')] \propto e^{(1-r)\alpha_0} ~~~~~~\forall c' \neq c, \\ \Pr[\krr(b,c)) = (b', c')] \propto 1 ~~~~~~ \forall b'\neq b, c'.
\end{gather*}
Using this mechanism, the higher-order terms of Eq.~\eqref{eq:emd-bound} will approach $0$ with $r$, as follows:

\begin{restatable}[]{thm}{krracc}\label{thm:krr-utility}
    For the metric space $\calX = \calB \times \calC$, \FreqEstLocal{} with the mechanism $\calA = \krr_{\alpha_0}$ satisfies $(\alpha, \delta)$-$\emd$ DP in the local model and returns an estimator $\tH$ such that
    \begin{multline}\label{eq:krr-utility}
        \max_K\E[\emd(\tH, \tK)] \leq r \sqrt{\frac{st^3}{mn}}\left(\frac{e^{\alpha_0} + s}{e^{\alpha_0} - e^{(1-r)\alpha_0}}\right) + \sqrt{\frac{s^2 t^2}{mn}}\left( \frac{\sqrt{s + 2(e^{\alpha_0} - 1)}}{e^{\alpha_0} + (t-1)e^{(1-r)\alpha_0}-t}\right),
     \end{multline}
    where $\alpha_0$ is defined in Eq.~\eqref{eq:eps-amp}.
\end{restatable}
\noindent
\textbf{Remarks.}
Specifically, for our choice of $\alpha = \frac{\epsilon}{r}$, we have
\[
    \max_{K} \E[\emd(\tK, \tH)] \leq 4\sqrt{\tfrac{k^3}{mn}} + 64\tfrac{\sqrt{k^3}}{\alpha\sqrt{n}}\sqrt{\ln(4m \exp(\alpha)/\delta)}.
\]

Similar to Lemma~\ref{lem:local-hist-utility-baseline}, the $\sqrt{\frac{k^{3}}{mn}}$ term is the cost of sampling. The $r \frac{\sqrt{k^3}}{\epsilon \sqrt{n}}$ term is the cost of privacy, and it dominates when $\alpha \leq \sqrt{m}$. 
We will compare Theorem~\ref{thm:krr-utility} with Lemma~\ref{lem:local-hist-utility-baseline} when $k, \epsilon, \alpha < \sqrt{m}$; then, the cost of privacy dominates. Specifically, the cost of Lemma~\ref{lem:local-hist-utility-baseline} is $O(\sqrt{\frac{k^2 \ln (m/\delta)}{n \epsilon^2}})$, and the cost of Theorem~\ref{thm:krr-utility} is $O(\sqrt{\frac{k^3}{\alpha^2 n} \max \{\ln(\frac{m}{\delta}), \alpha \}})$. Given $\epsilon$, the error will be smaller if
\[
    \alpha > \begin{cases} \epsilon \sqrt{k} & \epsilon < \frac{1}{\sqrt{k}} \ln(\frac{m}{\delta}) \\ \epsilon^2 \frac{k}{\ln(m/\delta)} & \text{otherwise} \end{cases}
\]
i.e. if there is a \emph{gap} between $\alpha, \epsilon$ of size at least $\sqrt{k}$. This is possible if $k \ll \frac{1}{r}$, and for these instances $\emd$ DP offers better utility than user-level DP.
In Theorem~\ref{thm:krr-utility}, the super-linear factor of $k^{3/2}$ comes from the fact that the $k$-RR is suboptimal in terms of $k$~\citep{acharya2018communication}. 


\subsubsection{Algorithms in the Central Model}
The Laplace mechanism has been shown to be optimal for many instances of frequency estimation~\citep{dwork2014algorithmic}. To attain user-level privacy, the baseline Laplace mechanism releases, for each $x \in \calX$, the values $F_x= \tK_G(x) + Y$, where $Y \sim Lap(\frac{1}{n\epsilon})$. The distribution function $\tH$ is then the normalization of $\langle F_x : x \in \calX \rangle$. This gives us the following guarantee.
\begin{lemma}\label{lem:lap-freq}
    For the metric space $\calX = \calB \times \calC$, the Laplace mechanism described above satisfies $(\epsilon, 0)$-user level DP, and produces an estimate $\tH$ such that
    \[
        \max_{K} \E[\emd(\tK, \tH)] \leq O\left(\tfrac{k}{n \epsilon}\right).
    \]
\end{lemma}

Again, this utility does not depend on $m$, since each user contributes $\frac{1}{n}$ fraction of the whole dataset which is independent of $m$. Consistent with central DP, the error decreases with $\frac{1}{n}$, which is much faster than the $\frac{1}{\sqrt{n}}$ in the local model. 

It is possible to adapt \FreqEstLocal{} to bounded central $\emd$-DP by simply pretending to be one user who holds the global dataset $K_G$. The privacy analysis of Theorem~\ref{thm:eps-amp}, and the utility analysis of Theorem~\ref{thm:krr-utility} may be combined for the following corollary:

\begin{restatable}[]{coro}{krracccent}\label{coro:krr-utility-central}
    For the metric space $\calX = \calB \times \calC$, \FreqEstLocal{} with $\calA = \krr_{\alpha_0}$ with $\alpha_0$ given in Eq.~\eqref{eq:eps-amp-central} satisfies $(\alpha, \delta)$-$\emd$ DP in the central model and returns an estimator $\tH$ with error given in Eq.~\eqref{eq:krr-utility}.
\end{restatable}
\noindent\textbf{Remarks.}
In particular
    \[
        \max_{K}\E[\emd(\tH, \tK)] \leq 4\tfrac{\sqrt{k^3}}{\sqrt{mn}} + 64 \tfrac{\sqrt{k^3}}{\alpha n} \sqrt{\ln(4m \exp(\alpha) / \delta)}.
    \]
The same sampling error is present, but the cost of privacy is reduced from a $\frac{1}{\sqrt{n}}$ dependence in Theorem~\ref{thm:krr-utility} to just $\frac{1}{n}$. To compare just the cost of privacy in Corollary~\ref{thm:krr-utility} to Lemma~\ref{lem:lap-freq}, we will assume we are in the regime $n \leq \frac m \alpha$.  Then, the cost in Corollary~\ref{coro:krr-utility-central} is $O(\frac{\sqrt{k^3}}{\alpha n})\sqrt{\max\{\ln(\frac{m}{\delta}), \alpha \}}$. The error of Corollary~\ref{coro:krr-utility-central} is less when
\[
\alpha \geq \begin{cases} \epsilon \sqrt{k \ln(\frac{m}{\delta})} & \epsilon \leq \sqrt{\frac{\ln(m / \delta)}{k}} \\ \epsilon^2 k & \text{otherwise} \end{cases}
\]
Thus, the utility is improved when $\alpha$ is bigger than $\epsilon$ by a factor of at least $\sqrt{k}$, which is achieved when $k \ll \frac{1}{r}$. One final advantage of Corollary~\ref{coro:krr-utility-central} is that it may be implemented in the shuffle model of DP, which requires less trust than the central model. This parallels prior results of the shuffle model of DP~\citep{feldman2022hiding}.
\\\noindent \add{\textbf{Open Questions.} For linear queries, one open question is whether it is possible to obtain error competitive with user-level DP under $(\alpha, 0)$-$\emd$-DP. For frequency estimation, an immediate open question is whether it is possible to reduce the super-linear dependence $k$ to one that matches that of Lemma~\ref{lem:local-hist-utility-baseline}, and whether an improvement in error can be made when $\alpha \leq \epsilon^2 k$.}

\section{Related Work}\label{sec:rel-work}

\textbf{Item-level DP.} DP was originally considered at the item-level~\citep{dwork2006differential}. 
Relevant to our setting are results in distribution estimation~\citep{hay2009boosting, xu2013differentially,suresh2019differentially}; these results study more complex estimation problems than frequency. We also consider linear query release~\citep{hardt2009geometry, bhaskara2012unconditional, nikolov2013geometry, blum2013learning, li2015matrix}. The mechanism in~\cite{hardt2009geometry} is often optimal and easy to adapt to our setting; we compare our algorithms with it.
\\\textbf{User-level DP.} User-level privacy is gaining increasing interest~\citep{amin2019bounding, narayanan2022tight, bassily2023user, levy2021learning, liu2020learning, cummings2022mean}. The most relevant work to ours involves user-level private mean estimation~\citep{cummings2022mean} and histogram estimation~\citep{liu2023algorithms, acharya2023discrete}, though these problems are more complex than those we study. Another related area is deciding the amount of data to collect from each user when users have varying amounts of data~\citep{amin2019bounding, liu2023algorithms, cummings2022mean}, which relates to our unbounded DP setup. These techniques apply to more specialized settings than our general blackbox reduction and are not immediately comparable.
\\\textbf{Local DP.} The results most relevant to our work in local DP are locally-private linear query release ~\cite{duchi2013local, bassily2019linear} and distribution estimation~\citep{duchi2013local, kairouz2016discrete, acharya2018communication, chen2020breaking, acharya2023discrete}. We directly compare our work to the optimal algorithms in~\cite{bassily2019linear} and~\cite{chen2020breaking} for our problems, which can be adapted to user-level DP easily. The other related line of work is privacy amplification via shuffling~\citep{erlingsson2019amplification, girgis2021renyi, feldman2022hiding}. We extend the state-of-the-art analysis in~\cite{feldman2022hiding} to general metric DP.
\\
\textbf{Metric DP.}
Metric DP was first proposed in~\cite{chatzikokolakis2013broadening} in the central model. In the local model, this has led to work on releasing numeric data~\citep{RoyChowdhury2022}, location data~\citep{andres2013geo, bordenabe_optimal_2014, chatzikokolakis2015constructing,weggenmann2021differential} and text~\citep{feyisetan2019leveraging,feyisetan2020privacy,feyisetan2021private,imola2022balancing}. Unlike these works, we consider privacy in a general metric space. The most related work is that of~\cite{fernandes2019generalised}, which proposes metric DP based on $\emd$ for releasing text embeddings. As explained in the introduction, we consider a more general setting than~\cite{fernandes2019generalised}.

\add{\vspace{-0.2cm}\section{Interpretation of $\emd$-DP}\label{sec:interpret}
In this section, we elaborate on how to interpret the $\emd$-DP guarantee (and metric DP in general) relative to standard DP. 
We start by discussing the advantages offered by $\emd$-DP. As discussed in Section \ref{sec:intro}, the primary benefit of $\emd$-DP is that it offers a more fine-grained and nuanced privacy definition compared to standard DP. This results in a more flexible privacy-utility trade-off that is better suited than standard DP for many real-world applications. In addition, $\emd$-DP unlocks new proof techniques that may also be applicable to standard DP. Specifically, the $\emd$ metric introduces couplings that need to be explicitly addressed in privacy analysis. For instance, in Section \ref{sec:reduct} we showed that sampling with replacement is a smooth projection by explicitly viewing sampling from two datasets in terms of a coupling between
them. While standard DP privacy analysis often implicitly uses couplings, we believe that some of our proof techniques for explicitly handling general couplings could also  be beneficial in the context of standard DP. 

Next, let us understand the technical relation between metric DP and standard DP. Metric DP is essentially a relaxation of standard DP. Any mechanism that satisfies metric DP (user-level or item-level) also satisfies standard DP, albeit with a potentially higher privacy parameter. In particular, any mechanism satisfying $(\alpha, \delta)$-$d_\calX$-DP also satisfies $(\alpha \cdot d_{max}, \delta)$-DP, where $d_{max}$ is the maximum $d_\calX$ distance between any two pairs of inputs. Conversely, any mechanism that satisfies $(\epsilon, \delta)$-DP also satisfies $(\frac{\epsilon}{d_{min}},\delta)$-$d_\calX$-D,P where $d_{min}$ is the minimum $d_\calX$ distance between any two pairs of inputs. Hence, although in theory one can translate between these two privacy guarantees, the translation is very loose. Tightly analyzing the privacy parameter under metric DP (whether $\emd$ or otherwise) for an arbitrary mechanism that satisfies standard DP is non-trivial and there is no one-size-fits-all method to do so.  For instance, Algorithm \ref{alg:lap} can be instantiated via the Laplace and Gaussian mechanism -- both classic standard DP mechanisms -- under some conditions. However, as discussed in Section \ref{sec:lin-avg}, a more complex analysis is required to evaluate privacy under $\emd$-DP. Additionally, in terms of mechanism design, a mechanism optimized for $\emd$-DP might not be ideal for standard DP and vice-versa. For instance, our proposed Algorithm \ref{alg:hist-local} for performing frequency estimation may not work well under standard DP (i.e., have high privacy parameters).

In what follows, we outline three concrete scenarios, where a practitioner should prefer $\emd$-DP over standard DP. First,  if the practitioner has a prior on the sensitive data indicating that the distance between any two user's data is overwhelmingly likely to be $<1$ (assuming all distances are normalized), then $\emd$-DP is clearly the better choice. This is because such a prior makes the worst-case scenario of antipodal data pairs—where two users' data are completely dissimilar (the case that standard DP safeguards against)—highly unlikely in practice.  For instance, this scenario may arise in the context of location data when the data corresponds to location information of employees of the same firm. In this case, weekday locations will be the same across all users, leading to small pairwise $\emd$ distances. Second,  a practitioner should opt for $\emd$-DP when it captures a more realistic privacy semantics of the underlying data.  Although ideally, we would like to prevent any data leakage about an individual, this is unfortunately not feasible in practice due to the  vast amount of auxiliary information already publicly available about every individual. For instance, most people's occupations are publicly available on social media profiles. Returning to our location data example, data collected over a month would reveal routine patterns, such as a person's workplace, which is already public and hence doesn't require protection. Rather, the more sensitive information is short-term location data gathered over say the course of a day (which might reveal non-routine visit to a friend or hospital).  In such a scenario, $\emd$-DP would offer a better privacy-utility trade-off with more realistic privacy guarantees than standard DP.
Third, $\emd$-DP may be preferable if the practitioner is restricted to work within a low privacy parameter regime (for instance, due to some government guideline). This is because for the same privacy parameter (i.e., $\alpha=\epsilon$), $\emd$-DP can offer a stronger privacy guarantee than standard DP while maintaining the same utility for certain queries, such as linear queries (Section \ref{sec:lin-avg}).

Finally, we conclude with some caveats regarding the use of metric DP. Metric DP assigns varying levels of sensitivity to different neighboring pairs. Specifically, smaller changes between neighboring pairs are considered more sensitive and are therefore protected with a higher privacy guarantee. However, this approach may not be suitable for all contexts. For instance, in the case of medical records, where the data between individuals can be vastly different, standard DP may offer better privacy protection. When adopting metric DP, it is crucial for practitioners to clearly define what is considered sensitive and what is not, and to engage in discussions about whether these definitions align with acceptable privacy semantics. This transparency enables users to make informed decisions about whether the privacy guarantees provided meet their needs.

}

\section{Conclusion}
\add{We have proposed metric DP at the user level using the earth-mover's distance, $\emd$.} This captures both the magnitude and structural aspects of changes in the data, resulting in a tailored privacy semantic. We have designed two novel privacy mechanisms under $\emd$-DP which improves the utility over standard DP. Additionally, we have shown that general (unbounded) $\emd$-DP can be reduced to the simpler case (bounded) where all users have the same amount of data. Finally, we have demonstrated that $\emd$-DP, when tailored to the application,  can offer improved utility over standard DP.
\bibliographystyle{ACM-Reference-Format}
\bibliography{main}

\appendix

\appendix
\section{Omitted Technical Details}

An alternative characterization of differential privacy is through the hockey-stick divergence~\citep{barthe2013beyond}. For probability distributions $P, Q$ defined on a space $\calY$, this is given by the following:
\begin{defn}\label{def:hockey}
    Let $\epsilon, \delta > 0$, and let $P, Q$ be distributions on a space $\calY$. The Hockey Stick Divergence is given by
    \[
    D_{e^\epsilon}(P \|Q) = \int_{\calY} \max\left\{\frac{P(y)}{Q(y)} - e^{\epsilon}, 0 \right\} Q(y)dy.
    \]
\end{defn}
It is easy to show that $D_{e^\epsilon}(M(K) \| M(K')) \leq \delta$ implies~\eqref{eq:dp}, so Definition~\ref{def:hockey} provides an alternative way to prove privacy.

Definition~\ref{def:hockey} satisfies a number of useful properties. First, because it is an $f$-divergence~\citep{csiszar1975divergence}, it satisfies the \emph{data-processing inequality}: for any function $f$, we have
\[
    D_{e^\epsilon}(f(P) \| f(Q)) \leq D_{e^\epsilon}(f(P) \| f(Q)).
\]
This property is used to show that DP is invariant to post-processing by any function $f$.
The second property, again holding for all $f$-divergences, is \emph{convexity}. This states that for two pairs of distributions $P_1, P_2, Q_1, Q_2 \in \Delta^\calY$ and a real number $\lambda \in [0, 1]$ we have 
\[
    D_{e^\epsilon}( \lambda P_1 + (1-\lambda) P_2 \| \lambda Q_1 + (1-\lambda) Q_2) \leq \lambda D_{e^\epsilon}(P_1 \| Q_1) + (1-\lambda) D_{e^\epsilon} (P_2 \| Q_2). 
\]
Stated in terms of couplings, we may generalize convexity as follows:
\begin{lemma}\label{lem:convex-coupling}
    Suppose $X,Y \in \calX$ are random variables with probability distributions $P_X, P_Y \in \Delta^\calX$. Suppose $\calM : \calX \rightarrow \calY$ is a randomized function. Then, for any coupling $C \in \calC(P_X, P_Y)$, we have
    \[
        D_{e^\epsilon}(\calM(X)\| \calM(Y)) \leq \E_{(x,y) \sim C}[D_{e^\epsilon}(\calM(x) \| \calM(y))].
    \]
\end{lemma}
\begin{proof}
    We may write 
    \begin{align*}
        \calM(X) &= \sum_{x \in \calX} P_X(x) \calM(x) = \sum_{x,y \in \calX} C(x,y) \calM(x) \\
        \calM(Y) &= \sum_{x,y \in \calX} C(x,y) \calM(y).
    \end{align*}
    Applying convexity, we have
    \[
        D_{e^\epsilon}(\calM(X) \| \calM(Y)) \leq \sum_{x,y \in \calX} C(x,y) D_{e^\epsilon}(\calM(x) \| \calM(y)),
    \]
    and the claim follows.
\end{proof}
Third, $D_{e^\epsilon}$ satisfies a ``weak'' triangle inequality (also known as group privacy):
\begin{lemma}\label{lem:group-priv}
    For distributions $P, Q, R$ on $\calY$, we have $D_{e^{\alpha + \beta}}(P\|R) \leq D_{e^\alpha}(P \|Q) + e^\alpha D_{e^\beta}(Q\|R)$.
\end{lemma}
\begin{proof}
    For any $P, Q, \epsilon$, we may view $D_{e^{\epsilon}}(P \| Q)$ through its dual form as
    \[
        D_{e^\epsilon}(P \| Q) = \sup_{Y \subseteq \calY} (P(Y) - e^\epsilon Q(Y)).
    \]
    Thus, let $Y^*$ denote the maximal set such that
    \begin{align*}
        D_{e^{\alpha+\beta}}(P \| R) = (P(Y^*) - e^{\alpha + \beta} R(Y^*)).
    \end{align*}
    We may rewrite this as 
    \begin{align*}
    D_{e^{\alpha+\beta}}(P \| R) &= (P(Y^*) - e^\alpha Q(Y^*)) + e^\alpha (Q(Y^*) -  e^{\beta} R(Y^*)) \\
    &\leq D_{e^\alpha}(P \| Q) + e^\alpha D_{e^\beta}(Q \| R),
    \end{align*}
    showing the claim.
\end{proof}

\section{Omitted Proofs from Section \ref{sec:mechanisms}}\label{app:mech-proofs}
\subsection{Proof of Theorem \ref{thm:sens}}
\sens*
For any two distributions $\tK, \tK'$, we have
\begin{align*}
    q_f(K) - q_f(K') &= \E_{x \sim \tK}[f(x)] - \E_{x \sim \tK'}[f(x)] \\
    &= \sum_{x \in \calX}f(x)\tK(x)  - \sum_{x \in \calX} f(y)\tK'(y).
\end{align*}
Let $C(x,y) = \{C_x(y)\}_{x \in \calX}$ be the minimum-transport coupling between $\tK$ and $\tK'$. By Definition~\ref{def:coupling}, we have $\tK'(y) = \sum_{x \in \calX} C(x,y)$, and $\emd(\tK, \tK') = \sum_{x, y \in \calX} d_\calX(x,y)C(x,y)$. Now, we write
\begin{align*}
    &\sum_{x \in \calX}f(x)\tK(x)- \sum_{y \in \calX} f(y)\tK'(y) \\
    &\qquad = \sum_{x \in \calX}f(x)\tK(x) - \sum_{y \in \calX} f(y) \sum_{x \in \calX} C(x,y) \\
    &\qquad = \sum_{x \in \calX} \left( f(x) - \sum_{y \in \calX} f(y) C_x(y)\right) \tK(x) \\
    &\qquad = \sum_{x \in \calX} \left( \sum_{y \in \calX} f(x)C_{x}(y) - \sum_{y \in \calX} f(y) C_{x}(y) \right) \tK(x) \\
    &\qquad = \sum_{x \in \calX} \sum_{y \in \calX} \left( f(x) - f(y)\right) C_{x}(y) \tK(x) \\
    &\qquad = \sum_{x, y \in \calX} \left( f(x) - f(y)\right) C(x,y).
\end{align*}
By the triangle inequality and the fact that $f$ is $\ell$-Lipschitz, we may write
\begin{align*}
    \|q_f(K) - q_f(K')\| &\leq \sum_{x, y \in \calX} \| f(x) - f(y)\| C(x,y) \\
    &\leq \sum_{x, y \in \calX} \ell d_{\calX}(x,y) C(x,y) \\
    &= \ell \emd(\tK, \tK').
\end{align*}
The last equation tells us that $\Delta_{\emd}(q_f) \leq \ell$.

\subsection{Proof of Lemma~\ref{lemma:lin}}
\lqpriv*
In the local model, by Theorem~\ref{thm:sens}, we have $\|q_f(K) - q_f(K')\| \leq \ell$. By adding noise drawn from $\Gamma(d, \frac{1}{\alpha})$, it is known this satisfies $(\alpha, 0)$-DP~\cite{hardt2009geometry}. In the bounded central setting, we have $\|q_f(K) - q_f(K')\| \leq \frac{\ell}{n}$, and thus we may add noise drawn from $\Gamma(d, \frac{1}{n \alpha})$.
\subsection{Proof of Theorem~\ref{thm:shuffle-metric-full}}\label{app:shuffle-metric-full}
\shufflepriv*
We will first assume the following lemma:
\begin{lemma}\label{lem:shuffle-metric-part}
    Suppose that $\calA$ is an $\alpha_0 d_\calX$-metric DP algorithm, where $d_\calX \leq 1$. Let $x_1^0, x_1^1, x_2, \ldots, x_m \in \calX$ be a set of inputs such that $d_\calX(x_1^0, x_1^1) \leq d$, and let $\delta > 0$ be a constant such that $\alpha_0 \leq \ln (\frac{m}{16 \ln (2/\delta)})$. Then, we have that
    \[
        D_{e^{\alpha}}(\mathsf{Shuffle}(\calA(x_1^0), \ldots, \calA(x_m)) \| \mathsf{Shuffle}(\calA(x_1^1), \calA(x_2), \ldots, \calA(x_m))) \leq \delta, 
    \]
    where 
    \[
        \alpha \leq \ln\left(1 + \frac{e^{\alpha_0 d}-1}{e^{\alpha_0 d}+1}\left(\frac{8 \sqrt{e^{\alpha_0} \ln (4 / \delta)}}{\sqrt{m}} + \frac{8 e^{\alpha_0}}{m} \right)\right).
    \]
\end{lemma}
To prove Theorem~\ref{thm:shuffle-metric-full}, 
let \[S(\mathbf{x}_i, \mathbf{x}_{m-i}') = \mathsf{Shuffle}(\calA(x_1), \ldots, \calA(x_i), \calA(x_{i+1}'), \ldots, \calA(x_m')).\] 
Let $k = \|v\|_0$, and WLOG suppose that $x_i = x_i'$ for $i > k$.
Our goal is to show that \[D_{e^\alpha}(S(\mathbf{x}_{k}, \mathbf{x}_{m-k}') \| S(\mathbf{x}_0, \mathbf{x}_{m}')) \leq \delta.\] 
By Lemma~\ref{lem:shuffle-metric-part}, we have for each $1 \leq i \leq k$ that \[D_{\exp(\alpha_i)}(S(\mathbf{x}_{i}, \mathbf{x}_{m-i}') \| S(\mathbf{x}_{i-1}, \mathbf{x}_{m-i+1}')) \leq \frac \delta {k},\] where $\alpha_i = f(d_\calX(x_i, x_i'))$, and 
\[
f(d) = \ln\left(1 + \frac{e^{\alpha_0 d}-1}{e^{\alpha_0 d}+1}\left(\frac{8 \sqrt{e^{\alpha_0} \ln (4k / \delta)}}{\sqrt{m}} + \frac{8 e^{\alpha_0}}{m} \right)\right).
\]
Applying Lemma~\ref{lem:group-priv} $k$ times, we see
\begin{align*}
    &D_{\exp(\alpha_1 + \cdots + \alpha_k)}(S(\mathbf{x}_{k}, \mathbf{x}_{m-k}') \| S(\mathbf{x}_0, \mathbf{x}_{m}')) \\
    &\leq D_{\exp(\alpha_k)}(S(\mathbf{x}_{k}, \mathbf{x}_{m-k}') \| S(\mathbf{x}_{k-1}, \mathbf{x}_{m-k+1}')) \\
    &+ e^{\alpha_k} D_{\exp(\alpha_{k-1})}(S(\mathbf{x}_{k-1}, \mathbf{x}_{m-k+1}') \| S(\mathbf{x}_{k-2}, \mathbf{x}_{m-k+2}')) \\
    &+ \cdots \\
    &+ e^{\alpha_2 + \cdots + \alpha_k} D_{\exp(\alpha_1)} (S(\mathbf{x}_{1}, \mathbf{x}_{m-1}') \| S(\mathbf{x}_{0}, \mathbf{x}_{m}')) \\
    &\leq e^{\alpha_1 + \cdots + \alpha_k} \sum_{i=1}^{k} D_{\exp(\alpha_i)} (S(\mathbf{x}_{i}, \mathbf{x}_{m-i}') \| S(\mathbf{x}_{i-1}, \mathbf{x}_{m-i+1}')) \\
    &\leq e^{\alpha_1 + \cdots + \alpha_k} \delta.
\end{align*}

We now show that $f(d)$ is a concave function on the domain $d \geq 0$; to do this we write $f(d) = \ln(1 + g(d) K)$, where $g(d) = \frac{e^{d} - 1}{e^{d} + 1}$ and $K > 0$ is a suitable constant. We will show that 
\[f''(d) = \frac{(1+g(d)K)g''(d)K - g'(d)^2K^2}{(1+g(d)K)^2} \leq 0.\] It suffices to show that $(1 + Kg(d))g''(d) \leq Kg'(d)^2$. We may write
    \begin{gather*}
        g(d) = 1 - \frac{2}{e^d+1} \\
        g'(d) = \frac{2e^d}{(e^d+1)^2} \\
        g''(d) = 2\frac{(e^d+1)^2e^d - 2e^d(e^d+1)e^d}{(e^d+1)^4} = 2\frac{e^d - e^{2d}}{(e^d+1)^3}.
    \end{gather*}
    Now, we have
    \begin{gather*}
        (1 + Kg(d))g''(d) \leq Kg'(d)^2 \\
        \Longleftrightarrow (1+K-\frac{2K}{e^d+1})2\frac{e^d - e^{2d}}{(e^d+1)^3} \leq K \frac{4e^{2d}}{(e^d+1)^4} \\
        \Longleftrightarrow((e^d+1)(K+1) - 2K)(1-e^{d}) \leq 2K e^{d} \\
        \Longleftrightarrow (Ke^d - K + e^d + 1)(1-e^d) \leq 2Ke^d \\
        \Longleftrightarrow 2Ke^d - K + 1 - Ke^{2d} - e^{2d} \leq 2Ke^{d} \\
        \Longleftrightarrow - K + 1 - Ke^{2d} - e^{2d} \leq 0 \\
    \end{gather*}
    We are done by observing that $1 - e^{2d} \leq 0$ when $d \geq 0$, and $-K - Ke^{2d} \leq 0$. Having shown concavity, the maximum value of $\sum_{i=1}^k f(d_\calX(x_i, x_i'))$ subject to the constraint $\sum_{i=1}^k d_\calX(x_i, x_i') = \|v\|_1$ is achieved when each $d_\calX(x_i, x_i')$ is equal to $\frac{\|v\|_1}{\|v\|_0}$. This gives us a bound of
    \[
        \|v\|_0 \ln\left(1 + \frac{e^{\alpha_0 \|v\|_1 / \|v\|_0}-1}{e^{\alpha_0 \|v\|_1 / \|v\|_0}+1}\left(\frac{8 \sqrt{e^{\alpha_0} \ln (4\|v\|_0 / \delta)}}{\sqrt{m}} + \frac{8 e^{\alpha_0}}{m} \right)\right),
    \]
    giving us the desired bound.
    
\subsection{Proof of Lemma~\ref{lem:shuffle-metric-part}}
This lemma can be viewed as a generalization of amplification by shuffling, which has the same setup but sets $d = 1$ and merely requires that $\calM$ satisfy $\epsilon$-local DP.
We generalize the approach of~\citet{feldman2022hiding}, starting with the the following preliminary claims.
\subsubsection{Preliminary Lemmas}
\begin{lemma}\label{lem:shuffle-post-lem}
(Generalization of Lemma 3.3 in~\citet{feldman2022hiding}). Let $X = \{x_1^0, x_1^1, x_2 \ldots, x_m\}$ be a set of indices, and for $x \in X$, let $R(x), Q(x)$ be two families of distributions and $\alpha \in [0,1], \beta \in [0, \frac{1}{2}]$ be coefficients such that
\begin{gather*}
    R(x_1^0) = (1-\alpha) Q(x_1^0) + \alpha Q(x_1^1) \\
    R(x_1^1) = \alpha Q(x_1^0) + (1-\alpha) Q(x_1^1) \\
    R(x_j) = \beta Q(x_1^0) + \beta Q(x_1^1) + (1-2\beta) Q(x_j) ~~~ \forall j \geq 2.
\end{gather*}
Then, there exists a post-processing mechanism $\calS$ such that 
\begin{align*}
\mathsf{Shuffle}(R(x_1^0), R(x_2), \ldots R(x_m)) &= \calS(A+1-\Delta, C-A+\Delta) \ \ \ \ \ \ \ \text{and} \\
\mathsf{Shuffle}(R(x_1^1), R(x_2), \ldots, R(x_m)) &= \calS(A+\Delta, C-A+1-\Delta),
\end{align*}
where $C \sim \text{Bin}(s-1, 2\beta)$, $A \sim \text{Bin}(C, \frac{1}{2})$, and $\Delta \sim \text{Bernoulli}(\alpha)$, and $\mathsf{Shuffle}$ is a uniformly random shuffle.
\end{lemma}
\begin{proof}
    Let $Y_1^0, Y_1^1, Y_2, \ldots, Y_m$ be distributions where $Y_1^b$ is defined over $\{0, 1\}$ and satisfies $Y_1^0(0) = 1-\alpha$ and $Y_1^1(1) = \alpha$ (with reversed probabilities if $b=1$), and $Y_j$ for $j \geq 2$ is defined over $\{0, 1, 2\}$ and satisfies $Y_j(0) = Y_j(1) = \beta$ and $Y_j(2) = 1-2\beta$. Let $F$ be a function returning a distribution satisfying
    \[
        F_j(v) = \begin{cases} Q(x_1^0) & v = 0 \\ Q(x_1^1) & v = 1 \\ Q(x_j) & \text{otherwise} \end{cases}
    \]
    
    Observe that by definition, the following probability distributions are equal for $b \in \{0,1\}$:
    \[
    R(x_1^b), R(x_2), \ldots, R(x_m) = F_1(Y_1^b), F_2(Y_2), \ldots, F_m(Y_m).
    \]
    Let $\textbf{0}(Y_1,\ldots, Y_m) $ denote the number of indices $j$ such that $Y_j = 0$, and define $\textbf{1}(Y_1,\ldots, Y_m)$ similarly. 
    We will show that there exists a post-processing function $\calS$ such that, for both $b \in \{0,1\}$, we have
    \begin{multline}\label{eq:post-process}
        \mathsf{Shuffle}(F_1(Y_1^b), F_2(Y_2), \ldots, F_m(Y_m)) \\ = \calS(\textbf{0}(Y_1^b,\ldots, Y_m), \textbf{1}(Y_1^b,\ldots, Y_m)).
    \end{multline}
    We will do this by conditioning on the event $E_{u,v}$ that 
    \[(\textbf{0}(Y_1^b, Y_2, \ldots, Y_m), \textbf{1}(Y_1^b, Y_2, \ldots, Y_m)) = (u,v),\] where $u,v \in \mathbb{N}$ satisfy $1 \leq u+v \leq m$. Now, define the vector $r = \mathsf{Shuffle}(F(Y_1), F_2(Y_2), \ldots, F_m(Y_m))$. Conditioned on $E_{u,v}$, $r$ is distributed according to the following process: 
    First, select a random partition $U \sqcup V \sqcup W = [m]$ such that $|U| = u$ and $|V| = v$, corresponding to the indices (after shuffling) where $Y_1^b, Y_2, \ldots, Y_m$ are equal to $0, 1$, or $2$. Next, let $\pi$ be a random injection from $W$ to $[m] \setminus 1$. Then, $r$ is distributed according to:
    \begin{gather}
        r(u) = Q(x_{1}^0) \ \ \ \forall u \in U \\
        r(v) = Q(x_{1}^1) \ \ \ \forall v \in V \\
        r(w) = Q(x_{\pi(w)}) \ \ \ \forall w \in W.
    \end{gather}
    The above process is independent of $\alpha, \beta$ given $E_{u,v}$. In particular, it does not care whether we replace $\alpha$ with $1-\alpha$, and thus it serves as our process $\calS$ satisfying~\eqref{eq:post-process} for both values of $b$.
    Having established this, it is easy to show that $\textbf{0}(Y_1^0, \ldots, Y_m) = A+1-\Delta$, $\textbf{1}(Y_1^0, \ldots, Y_m) = C-A+\Delta$ for $b=0$, and $\textbf{0}(Y_1^1, \ldots, Y_m) = A+\Delta$, $\textbf{1}(Y_1^1, \ldots, Y_m) = C-A+1-\Delta$ for $b=1$.
\end{proof}

Having reduced the shuffling problem to a divergence between two fixed probability distributions, we follow the method of~\cite{feldman2022hiding} to compute this divergence. We use the following two results:

\begin{lemma}\label{lem:bin-divergence}
    (Restatement of Lemma A.1 from~\cite{feldman2022hiding}):
    Suppose $p \geq \frac{16 \ln(2 / \delta)}{m}$, $C \sim Bin(m-1, p)$ and $A \sim Bin(C, \frac{1}{2})$. Define $P = (A+1, C-A)$ and $Q = (A, C-A+1)$. Then, $D_{e^{\epsilon}}(P \| Q) \leq \delta$, where 
    \[
    \epsilon = \ln\left(1 + \frac{8 \sqrt{\ln (4 / \delta)})}{\sqrt {pm}} + \frac{8}{pm}\right)
    \]
\end{lemma}

The next result, advanced joint convexity, originally appeared in the privacy amplification by sampling literature and can be used to improve the parameter $\epsilon$ when computing $D_{\alpha}(P \| Q)$ between two distributions which are nearly the same.
\begin{lemma} \label{lem:adv-conv}(Restatement of Theorem 2 from~\cite{balle2018privacy})
    Let $P,Q$ be probability distributions satisfying $P = \nu M + (1-\nu) N$ and $Q = \nu M' + (1-\nu) N$ for distributions $M,M',N$ and $\nu \in [0,1]$. Given $\alpha \geq 1$, define $\alpha' = 1 + \nu (\alpha-1)$ and $\beta = \frac{\alpha'}{\alpha}$. Then,
    \[
        D_{\alpha'}(P \| Q) \leq \nu D_\alpha (M \| (1-\beta) N + \beta M').
    \]
\end{lemma}

Finally, we require a result from local DP:
\begin{lemma}\label{lem:dp-decomp} (Restatement of Theorem 2.5 from~\cite{kairouz2015composition})
    Let $P, Q$ be two distributions and $\alpha \geq 1$ be a parameter such that $D_\alpha(P \| Q) = 0$. Then, there exist distributions $M, N$ such that
    \begin{gather*}
        P = \frac{\alpha}{\alpha+1}M + \frac{1}{\alpha+1} N \\
        Q = \frac{1}{\alpha+1} M + \frac{1}{\alpha+1} N.
    \end{gather*}
\end{lemma}
With these results in order, we are ready to complete the proof.

\subsubsection{Completing the proof of Lemma~\ref{lem:shuffle-metric-part}}
Using the definition of $d_\calX$-DP and the fact that $d_\calX \leq 1$, we have 
\begin{gather*}
    D_{\exp(\epsilon_0 d)}(\calA(x_1^0) \| \calA(x_1^1)) = 0 \\ D_{\exp(\epsilon_0)}(\calA(x_1^0) \| \calA(x_j)) = 0 ~~~ \forall j \geq 2 \\
    D_{\exp(\epsilon_0)}(\calA(x_1^1) \| \calA(x_j)) = 0 ~~~ \forall j \geq 2.
\end{gather*}
Applying Lemma~\ref{lem:dp-decomp} to the first equation, we obtain
\begin{gather}
    \calA(x_1^0) = (1-\beta) Q(x_1^0) + \beta Q(x_1^1) \label{eq:first-decomp-1} \\
    \calA(x_1^1) = \beta Q(x_1^0) + (1-\beta) Q(x_1^1)  \label{eq:first-decomp-2}
\end{gather}
where $\beta = \frac{1}{1+\exp(\epsilon_0 d)}$. Applying the lemma to the second and third sets of equations, we obtain
\begin{gather}
    \calA(x_1^0) = (1-\gamma) R(x_1^0, x_j) + \gamma R'(x_1^0, x_j) ~~~ \forall j \geq 2\label{eq:pair-decomp-1}\\
    \calA(x_j) = \gamma R(x_1^0, x_j) + (1-\gamma) R'(x_1^0, x_j) ~~~ \forall j \geq 2\label{eq:pair-decomp-2}\\
        \calA(x_1^1) = (1-\gamma) R(x_1^1, x_j) + \gamma R'(x_1^1, x_j) ~~~ \forall j \geq 2\label{eq:pair-decomp-3}\\
    \calA(x_j) = \gamma R(x_1^1, x_j) + (1-\gamma) R'(x_1^1, x_j) ~~~ \forall j \geq 2.\label{eq:pair-decomp-4}
\end{gather}
where $\gamma = \frac{1}{1 + \exp(\epsilon_0)}$. Subtracting~\ref{eq:pair-decomp-1} and~\ref{eq:pair-decomp-2}, we obtain that
\begin{gather}
    \calA(x_j) = \frac{\gamma}{1-\gamma} \calA(x_1^0) + \frac{1-2\gamma}{1-\gamma}R'(x_1^0, x_j) ~~~ \forall j \geq 2, \label{eq:pair-decomp-5}
\end{gather}
and likewise~\ref{eq:pair-decomp-3} and~\ref{eq:pair-decomp-4} imply
\begin{gather}
    \calA(x_j) = \frac{\gamma}{1-\gamma} \calA(x_1^1) + \frac{1-2\gamma}{1-\gamma}R'(x_1^1, x_j) ~~~ \forall j \geq 2. \label{eq:pair-decomp-6}
\end{gather}
Taking the average of~\ref{eq:pair-decomp-5} and~\ref{eq:pair-decomp-6}, we obtain 
\begin{gather}
    \calA(x_j) = \frac{\gamma}{2(1-\gamma)} \calA(x_1^0) + \frac{\gamma}{2(1-\gamma)} \calA(x_1^1) + \frac{1-2\gamma}{1-\gamma} Q(x_j) ~~~ \forall j \geq 2,
\end{gather}
where $Q(x_j) = \frac{1}{2} R'(x_1^0, x_j) + \frac{1}{2} R'(x_1^1, x_j)$. 
Now, equations~\ref{eq:first-decomp-1} and~\ref{eq:first-decomp-2} imply that
\[
\calA(x_1^0) + \calA(x_1^1) = Q(x_1^0) + Q(x_1^1).
\]
This implies
\begin{gather}
    \calA(x_j) = \frac{\gamma}{2(1-\gamma)} Q(x_1^0) + \frac{\gamma}{2(1-\gamma)} Q(x_1^1) + \frac{1-2\gamma}{1-\gamma}Q(x_j) ~~~ \forall j \geq 2.
\end{gather}
Applying Lemma~\ref{lem:shuffle-post-lem}, there exists a function $S$ such that 
\begin{gather*}
    \mathsf{Shuffle}(\calA(x_1^0), \calA(x_2), \ldots, \calA(x_m)) = S(A+1-\Delta, C-A+\Delta)\\
    \mathsf{Shuffle}(\calA(x_1^1), \calA(x_2), \ldots, \calA(x_m)) = S(A+\Delta, C-A+1-\Delta),
\end{gather*}
where $C \sim Bin(m-1, \frac{\gamma}{1-\gamma}) = Bin(m-1, e^{-{\epsilon_0}})$, $A \sim Bin(C, \frac{1}{2})$, and $\Delta \sim Bernoulli(\beta)$.
By the post-processing inequality, we have for any $\alpha \geq 1$ that
\begin{multline*}
D_\alpha (\mathsf{Shuffle}(\calA(x_1^0), \calA(x_2), \ldots, \calA(x_s)) \| \mathsf{Shuffle}(\calA(x_1^1), \calA(x_2), \\ \ldots, \calA(x_s))) \leq 
D_{\alpha}((A+1-\Delta, C-A+\Delta) \| (A+\Delta, C-A+1-\Delta)).
\end{multline*}
Observe we can write 
\begin{gather*}
    (A+1-\Delta, C-A+\Delta) = (1-\beta) (A+1, C-A) + \beta (A, C-A+1) \\
    (A+\Delta, C-A+1-\Delta) = \beta (A+1, C-A) + (1-\beta) (A, C-A+1).
\end{gather*}
Define $X = (A+1, C-A)$ and $Y = (A, C-A+1)$. We can rewrite the above as
\begin{gather*}
    (A+1-\Delta, C-A+\Delta) = 2\beta \frac{X+Y}{2} + (1-2\beta)X \\
    (A+\Delta, C-A+1-\Delta) = 2\beta \frac{X+Y}{2} + (1-2\beta)Y.
\end{gather*}
Applying Lemma~\ref{lem:adv-conv}, we have
\begin{multline*}
D_{\alpha'}((A+1-\Delta, C-A+\Delta) \| (A+\Delta, C-A+1-\Delta)) \\ \leq (1-2\beta) D_\alpha(X \| (1-\eta) (\tfrac{X+Y}{2}) + \eta Y),
\end{multline*}
where $\alpha' = 1 + (1-2\beta)(\alpha-1)$ and $\eta = \frac{\alpha'}{\alpha}$.
By convexity, the RHS above is at most
\[
D_{\alpha'}((A+1-\Delta, C-A+\Delta) \| (A+\Delta, C-A+1-\Delta)) \leq (1-2\beta) D_\alpha(X\|Y).
\]
Now, we finally set $\alpha = 1 + \frac{8 \sqrt{\exp(-\epsilon_0)\ln (4/\delta)}}{\sqrt{m}} + \frac{8\exp(-\epsilon_0)}{m}$. Lemma~\ref{lem:bin-divergence} (using the assumption that $\epsilon_0 \leq \ln (\frac{m}{16 \ln (2/\delta)})$) implies $D_{\alpha}(X\|Y) \leq \delta$. From this, we obtain our desired result that 
\begin{multline*}
D_{\alpha'} (\mathsf{Shuffle}(\calA(x_1^0), \calA(x_2), \ldots, \calA(x_m)) \| \mathsf{Shuffle}(\calA(x_1^1), \calA(x_2), \\ \ldots, \calA(x_m))) \leq (1-2\beta)D_{\alpha}(X\|Y) \leq \delta,
\end{multline*}
where 
\[
\alpha' = 1 + \frac{e^{\epsilon_0 d} - 1}{e^{\epsilon_0 d}+1}\left( \frac{8 \sqrt{e^{\epsilon_0}\ln (4/\delta)}}{\sqrt{m}} + \frac{8e^{\epsilon_0}}{m}\right).
\]

\subsection{Proof of Theorem \ref{thm:eps-amp}}
\epsamp*

    First, consider the local model.
    Fix any two itemsets $K = \{x_1, \ldots, x_m\}$ and $K' = \{x_1, \ldots, x_m'\}$ such that $\emd(\tK, \tK') \leq w$. By Lemma~\ref{lem:bvn}, there exists a permutation $\pi : [m] \rightarrow [m]$ such that
    \[
        \sum_{i=1}^m d_\calX(x_i, x_{\pi(i)}') = mw.
    \]
    Let 
    \begin{gather}
        \tL = \mathrm{Shuffle}(\calA(x_1), \ldots, \calA(x_m)) \\
        \tL' = \mathrm{Shuffle}(\calA(x_{\pi(i)}'), \ldots, \calA(x_{\pi(m)}')).
    \end{gather}
    By Theorem~\ref{thm:shuffle-metric-full}, we know that $D_{\exp(\alpha(w))}(\tL\| \tL') \leq \delta e^{\alpha(w)}$, where $\alpha(w) = h(m; m, mw)$. The final privacy parameters for a fixed $w$ will be $\frac{\alpha(w)}w$ and $\delta e^{\alpha(w)}$; the worst-case privacy parameters are thus $\sup_{w \in [0,1]} \frac{\alpha(w)}{w}$ and $\sup_{w \in [0,1]} \delta e^{\alpha(w)}$. Since $\alpha(w)$ is an increasing function, the latter term reduces to $\delta e^{\alpha(w)}$. 

    In the bounded central model, the same logic applies, except that $\tL, \tL'$ have size $mn$, differ in only $m$ coordinates, and     
    \[
        \sum_{i=1}^{mn} d_\calX(x_i, x_{\pi(i)}') = mw.
    \]
    We apply Theorem~\ref{thm:shuffle-metric-full} to obtain $D_{\exp(\alpha(w))}(\tL\| \tL') \leq \delta e^{\alpha(w)}$, where $\alpha(w) = h(mn; m, mw)$, and we complete the proof similarly.

\section{Omitted Proofs from Section~\ref{sec:reduct}}\label{app:reduct-proofs}
\subsection{Proof of Lemma \ref{lem:coupling-dist}}
\smoothproj*
    For $i = 1, \ldots, s$, define $X_i = d_\calX(x_i, y_i)$, and observe that $\emd(\tL, \tL') \leq \frac{1}{s}(X_1 + \cdots + X_s)$. Now, let $\mu$ denote $\emd(\tK, \tK')$. Observe each $X_i$ is i.i.d. and satisfies $\E[X_i] = 
    \mu$ and $0 \leq X_i \leq 1$. Due to the last two facts, we have $\E[X_i^2] \leq \mu$. By Bernstein's inequality, we have, for all $t \geq 0$,
    \[
        \Pr\left[X_1 + \cdots + X_s - s\mu  \geq t\right] \leq e^{-t^2 / 2(v + bt / 3)},
    \]
    where $v = \sum_{i=1}^s \E[X_i^2] \leq s\mu$ and $b = 1$. By setting 
    \[
        t = \max\{ \sqrt{4 s \mu \ln (1/\delta)}, \tfrac{4}{3} \ln (1/\delta) \},
    \]
    we ensure that the probability is at most $\delta$. We have
    \[
        s \mu + t \leq s\mu + 2 \sqrt{s \mu \ln(1/\delta)} + \tfrac{4}{3} \ln(1/\delta) \leq (1+\sqrt{2})s\mu + (\tfrac 4 3 + \sqrt{2}) \ln \tfrac{1}{\delta}. 
    \]
    Finally,
    \begin{multline*}
        \Pr[\emd(\tL, \tL') \geq (1+\sqrt{2}) \mu + \tfrac{3}{s} \ln \tfrac{1}{\delta}] \leq \\ \Pr[X_1 + \cdots + X_s \geq (1+\sqrt{2}) s\mu + 3 \ln \tfrac{1}{\delta}] \leq \delta.
    \end{multline*}

\subsection{Proof of Theorem \ref{thm:emd-dp-central}}
\reductpriv*
First, we will consider the local model.
Let $K, K'$ denote two datasets such that $\emd(\tK, \tK') \leq r$. Let $L$, $L'$ denote the set of $s$ samples when $K$ (resp. $K'$) is used. Our goal is to show that $D_{\exp(\epsilon)}(\calM(L) \| \calM(L')) \leq \delta$. Observe we may define the objects $\mathbf{L}, \mathbf{L}' \in \Delta^{\calX^{s}}$ to be the probability distributions of $L, L'$ (which lie in $\calX^{s}$). By Lemma~\ref{lem:convex-coupling}, for any coupling $C \in \calC(\mathbf{L}, \mathbf{L}')$, we have
\begin{align*}
D_{\exp(\epsilon)}(\calM(L) \| \calM(L'))&\leq \E_{(L, L') \sim C} [ D_{\exp(\epsilon)}(\calM(L) \| \calM(L')) ].
\end{align*}
 Let $A$ denote the event that we have $\emd(\tL, \tL') \leq (1+\sqrt{2})r + \frac{3}{s} \ln \tfrac{1}{\delta}$. When $A$ holds, then $D_{\exp(\epsilon)} (\calM(L) \| \calM(L')) \leq \delta$ by assumption. When this does not hold, then trivially $D_{\exp(\epsilon)} (\calM(L) \| \calM(L')) \leq 1$. Conditioning on the above expectation, we have
\begin{align*}
\E_{(L, L') \sim C} [ D_{\exp(\epsilon)}(\calM(L) \| \calM(L'))] &\leq \delta \Pr[A] + \Pr[\overline{A}] \\
&\leq \delta + \Pr[\overline{A}].
\end{align*}

Now, let $C^* \in \Delta^{\calX \times \calX}$ denote the optimal coupling between $\tK, \tK'$. We will take $C = (C^*)^s \in \Delta^{\calX^{s} \times \calX^{s}}$, the $s$-fold Kronecker product of $C^*$. Observe this is indeed a coupling between $\mathbf{L}, \mathbf{L'}$, and each coordinate of $(L, L') \sim C$ is simply a sample from $C^*$. Thus, the event $A$ above is equivalent to
\[
    \Pr[A] = \Pr_{(L, L') \sim (C^*)^s}[\emd(\tL, \tL') \leq (1+\sqrt{2})r + \frac{3}{s} \ln \frac{1}{\delta}],
\]
where the notation $(L, L') \sim (C^*)^s$ indicates that $L = \{x_1, \ldots, x_s\}$ and $L = \{y_1, \ldots, y_s\}$, and each $(x_i, y_i) \sim C^*$.
By Lemma~\ref{lem:coupling-dist}, we know that $\Pr[A] \geq 1-\delta$, and thus the above expectation is at most $2\delta$. This proof may be generalized easily to the central model.

\section{Omitted Proofs from Section~\ref{sec:utility}}\label{app:uti-proofs}

\subsection{Proof of Lemma~\ref{lem:lq-local}}
\lqacc*
As the sensitivity of the query is bounded by $\|F\|_2$, is easy to show (e.g.~\cite{dwork2014algorithmic}) that adding $d$-dimensional Gaussian noise with width $\|F\|_2\frac{r\sqrt{1.25 \ln \frac{1}{\delta}}}{\alpha}$ in each coordinate will satisfy $(\frac{\alpha}{r}, \delta)$ local $\emd$-DP. The standard deviation in each coordinate of $\hat{q}$ is thus $\|F\|_2\frac{r\sqrt{1.25 \ln \frac{1}{\delta}}}{\alpha \sqrt{n}}$, and this gives the desired expected error.

\subsection{Proof of Theorem~\ref{thm:general-hist-utility}}
\freqacc*
First, we will introduce notation.
    For a cluster label $b \in \calB$, let $\calX[b] \subseteq \calX$ denote the elements of $\calX$ in cluster $b$. Define $\tF[b] \in \R^{\calB \times \calC}$ to be the indices of $\tF$ in $\calX[b]$ (so that indices outside $\calX[b]$ are zeroed out). Define $\tK[b]$ similarly, and observe that $\tF[b], \tK[b]$ are not normalized. 
    
    For any estimate $\tF$, consider the following transportation plan from $\tF$ to $\tK$: For each $b \in \calB$, transfer $\tF[b]$ to $\tK[b]$ arbitrarily, and put any excess weight in the bin $(b, c')$ for an arbitrary $c' \in \calC$. The cost incurred by this is at most $r\|\tF[b] - \tK[b]\|_1 + r|\mu(\tF[b]) - \mu(\tK[b])|$, where $\mu(\cdot)$ denotes total mass of its argument. Finally, equalize the weights in the coordinates $\{(b, c') : b \in \calB\}$. The cost incurred for this step is at most $(1-r)\sum_{b \in \calB} |\mu(\tF[b]) - \mu(\tK[b])|$. Thus, the total cost is
    \begin{multline*}
        \sum_{b \in \calB} r\|\tF[b] - \tK[b]\|_1 + |\mu(\tF[b]) - \mu(\tK[b])| \\ = r \|\tF - \tK\|_1 + \sum_{b \in \calB} |\mu(\tF[b]) - \mu(\tK[b])|.
    \end{multline*}
    Observe that the term $\sum_{b \in \calB} |\mu(\tF[b]) - \mu(\tK[b])|$ is simply the $\ell_1$ distance between $\tF P$ and $\tK P$, where $P \in \R^{(\calB \times \calC) \times \calB}$ is the matrix that maps a vector to its sum along each coordinate in $\calB$. Thus, we may form the the upper bound
    \begin{align}
        \E[\emd(\tF, \tK)] &\leq r \E[\|\tF - \tK\|_1] + \E[\|(\tF - \tK)P\|_1] \nonumber \\
        &\leq r \E[\sqrt{st} \|\tF - \tK\|_2] + \E[\sqrt{s} \|(\tF - \tK)P\|_2] \nonumber \\
        &\leq r \sqrt{st\E[ \|\tF - \tK\|_2^2]} + \sqrt{s\E[ \|(\tF - \tK)P\|_2^2]}. \label{eq:emd-ub-2}
    \end{align}

    Now, we will bound \eqref{eq:emd-ub-2} given this estimator. In the following, let $A_x$ denote the $x$th row of the matrix $A$. Observe that
    \begin{align*}
        \tF - \tK &= 
        \frac{1}{mn} \sum_{i=1}^{mn} z_i B - \tK AB \\
        &= \frac{1}{mn} \sum_{i=1}^{mn} z_i B - \frac{1}{mn} \sum_{i=1}^{mn} e_{k_i} A B \\
        &= \frac{1}{mn} \sum_{i=1}^{mn} z_i B - \frac{1}{mn} \sum_{i=1}^{mn} A_{k_i} B \\
        &= \frac{1}{mn} \sum_{i=1}^{mn} (z_i-A_{k_i}) B \\
    \end{align*}
    Define $w_i = z_i - A_{k_i}$, and notice that $\E[w_i] = \E[z_i] - A_{k_i} = 0$. Thus,
    \begin{align*}
        \E[\|\tF - \tK\|_2^2] &= \E[(\tF - \tK) (\tF - \tK)^T] \\
        &= \left(\frac{1}{mn}\right)^2 \E\left[\left(\sum_{i=1}^{mn} w_i B \right) \left(\sum_{i=1}^{mn} B^T w_i^T\right)\right] \\
        &= \left(\frac{1}{mn}\right)^2 \sum_{i,j=1}^{mn} \E[w_i B B^T w_j^T] \\
        &= \left(\frac{1}{mn}\right)^2 \sum_{i=1}^{mn} \E[w_i B B^T w_i^T],
    \end{align*}
    where the last step holds because the $w_i$ are independent. Now, we have
    \begin{align*}
        \E[w_iB B^T w_i^T] &= \E[z_i B B^T z_i^T] - \E[A_{k_i} B B^TA_{k_i}^T] \\
        &= \E[z_i B B^T z_i^T] - e_{k_i}e_{k_i}^T \\
        &\leq \|B^T\|_{1,2}^2 - 1.
    \end{align*}
    Putting it all together, we have
    \[
        \E[\|\tF - \tK\|_2^2] \leq \frac{\|B^T\|_{1,2}^2 - 1}{mn} 
    \]
    To control the term $\|(\tF - \tK)P\|_2^2$ in~\eqref{eq:emd-ub-2}, using similar steps, we may write
    \[
        \E[\|(\tF - \tK)P\|_2^2] \leq \left(\frac{1}{mn}\right)^2 \sum_{i=1}^{mn} \E[w_i B P P^T B^T w_i^T].
    \]
    Similarly, for any $i$ we have
    \begin{align*}
        \E[w_i B P P^T B^T w_i^T] &= \E[z_i B P P^T B^T z_i^T] - \E[A_{k_i}B P P^TBA_{k_i}^T] \\
        &\leq \|P^TB^T\|_{1,2}^2 - 1,
    \end{align*}
    and this implies 
    \[
        \E[\|(\tF - \tK)P\|_2^2] \leq \frac{\|P^TB^T\|_{1,2}^2 - 1}{mn}.
    \]
    Substituting into~\eqref{eq:emd-ub-2}, we obtain the desired bound.

\subsection{Proof of Theorem~\ref{thm:krr-utility}}\label{app:proof-krr-utility}
\krracc*
    For positive constants $a, b, c$, the matrix $A$ is given by
    \[
        A = a I_\calX + (b I_\calB + c \textbf{1}_\calB) \otimes \textbf{1}_{\calC},
    \]
    where
    \begin{align*}
        a &= \frac{e^{\alpha_0} - e^{(1-r)\alpha_0}}{e^{\alpha_0} + (t-1)e^{(1-r)\alpha_0} + (s-1)t} \\
        b &= \frac{e^{(1-r)\alpha_0} - 1}{e^{\alpha_0} + (t-1)e^{(1-r)\alpha_0} + (s-1)t} \\
        c &= \frac{1}{e^{\alpha_0} + (t-1)e^{(1-r)\alpha_0} + (s-1)t}.
    \end{align*}
    The matrix $A$ is actually invertible, and 
    \[
        A^{-1} =a' I_\calX + (b' I_\calB + c' \textbf{1}_\calB) \otimes \textbf{1}_{\calC},
    \]
    where 
    \begin{align*}
        a' &= \frac{e^{\alpha_0} + (t-1)e^{(1-r)\alpha_0} + (s-1)t}{e^{\alpha_0} - e^{(1-r)\alpha_0}} \\
        b' &= -\frac{(e^{(1-r)\alpha_0}-1)(e^{\alpha_0} + (t-1)e^{(1-r)\alpha_0} + (s-1)t)}{(e^{\alpha_0} - e^{(1-r)\alpha_0})(e^{\alpha_0} + (t-1) e^{(1-r)\alpha_0}-t)} \\
        c' &= -\frac{1}{e^{\alpha_0} + (t-1) e^{(1-r)\alpha_0} - t}.
    \end{align*}
    It is easy to show the identity that $a' + tb' + stc' = 1$.
    Each row of $A^{-1}$ looks like one copy of $a' + b' + c'$, $t-1$ copies of $b' + c'$, and $(s-1)t$ copies of $c'$. Thus,
    \begin{align*}
        &\|(A^{-1})^T\|_{1 \rightarrow 2}^2-1 \\ &= (a' + b' + c')^2 + (t-1) (b' +  c')^2 + (s-1)t (c')^2 -1 \\
        &\;= (1-(t-1)b' - (st-1)c' )^2 + (t-1)(b')^2 \\
        &\qquad + 2(t-1) b'c' + (t-1)(c')^2 + (s-1)t (c')^2 - 1\\
        &\;= -2(t-1)b' - 2 (st-1)c' + 2(t-1)(st-1)c'b' \\
        &\qquad+ (t-1)^2(b')^2 + (st-1)^2 (c')^2 + (t-1)(b')^2 \\
        &\qquad + 2(t-1) b'c' + (t-1)(c')^2 + (s-1)t (c')^2 \\
        &\leq (tb')^2 + 2st^2 b' c' + (stc')^2 - 2tb' - 2stc' \\
        &\leq (tb' + stc')^2 - 2(tb' + stc') \\
        &= (a')^2 - 1.
    \end{align*}
    Substituting, we obtain
    \begin{align*}
        (a')^2-1 &\leq \left(\frac{te^{(1-r)\alpha_0} + (s-1)t}{e^{\alpha_0}-e^{(1-r)\alpha_0}}\right)^2 + 2 \left(\frac{te^{(1-r)\alpha_0} + (s-1)t}{e^{\alpha_0}-e^{(1-r)\alpha_0}}\right) \\
        &\leq \frac{t^2 e^{2\alpha_0} + 2(s-1)t^2e^{\alpha_0} + (s-1)^2t^2 + 2te^{2\alpha_0} + 2(s-1)te^{\alpha_0}}{(e^{\alpha_0} - e^{(1-r)\alpha_0})^2} \\
        &\leq \left(t \frac{e^{\alpha_0} + s}{e^{\alpha_0} - e^{(1-r)\alpha_0}}\right)^2
    \end{align*}
    Next, it's easy to see that
    \begin{align*}
        A^{-1} P &= \left(a' I_\calX + ((b' I_\calB + c'\textbf{1}_\calB) \otimes \textbf{1}_\calC)\right)(I_\calB \otimes 1_\calC) \\
        &= a' I_\calB \otimes 1_\calC + (b' I_\calB + c'\textbf{1}_\calB) \otimes t 1_\calC
    \end{align*}
    Each row of the latter consists of one copy of $a' + tb' + tc'$ and $s-1$ copies of $tc'$. This gives us
    \begin{align*}
    \|(A^{-1}P)^T\|_{1 \rightarrow 2}^2 - 1 &= (a' + tb' + tc')^2 + (s-1) (tc')^2 - 1 \\
    &= (1 - (s-1)tc')^2 + (s-1)(t c')^2 - 1 \\
    &= s(s-1)(t c')^2 - 2(s-1)(tc') \\
    &\leq (stc')^2 - 2(stc').
    \end{align*}
    Substituting, we obtain
    \begin{align*}
        (stc')^2 - 2(stc') &= \frac{st(st + 2(e^{\alpha_0} + (t-1) e^{(1-r)\alpha_0} - t))}{(e^{\alpha_0} + (t-1) e^{(1-r)\alpha_0} - t)^2}\\
        &\leq \frac{st^2(s + 2(e^{\alpha_0} - 1))}{(e^{\alpha_0} + (t-1) e^{(1-r)\alpha_0} - t)^2}
    \end{align*}
    Applying Theorem~\ref{thm:general-hist-utility}, we obtain
    \begin{align*}
        &\E[\emd(\tF, \tK)] \\
        &\;\leq r \sqrt{\frac{st((a')^2-1)}{mn}} + \sqrt{\frac{s^2t(st(c')^2 - 2c')}{mn}}  \\
        &\;\leq r \sqrt{\frac{st^3}{mn}}\left(\frac{e^{\alpha_0} + s}{e^{\alpha_0} - e^{(1-r)\alpha_0}}\right) + \sqrt{\frac{s^2 t^2}{mn}}\left( \frac{\sqrt{s + 2(e^{\alpha_0} - 1)}}{e^{\alpha_0} + (t-1)e^{(1-r)\alpha_0}-t}\right),
    \end{align*}
    finishing the claim.
To obtain an asymptotic bound (with budget $\alpha = \epsilon / r$), we plug in~\eqref{eq:eps-amp}, which says that we may set
\[ 
    \alpha_0 = \begin{cases} \frac{\alpha}{32\sqrt {m \ln(4m\exp(\alpha) / \delta)}} & \text{ if $\alpha \leq 32 \sqrt{m \ln (4m\exp(\alpha) / \delta)}$} \\ 2\ln \left( \frac{\epsilon}{16r\sqrt{m \ln(4m\exp(\alpha) / \delta)}} \right) & 32 r\sqrt{m \ln (4m\exp(\alpha) / \delta)} \leq \epsilon \leq rm \end{cases}.
\]
In the first case, we have
\begin{gather*}
    \frac{e^{\alpha_0} + s}{e^{\alpha_0} - e^{(1-r)\alpha_0}} \leq \frac{s}{r\alpha_0} \\
    \frac{\sqrt{s + 2(e^{\alpha_0} - 1)}}{e^{\alpha_0} + (t-1)e^{(1-r)\alpha_0}-t} \leq \frac{2\sqrt{s}}{\alpha_0 t},
\end{gather*}
and this implies
\begin{align*}
    \E[\emd(\tK, \tF)] &\leq \sqrt{\frac{s^3t^3}{mn}}\frac{1}{\alpha_0} + \sqrt{\frac{s^3t^2}{mn}}\frac{2}{\alpha_0} \\
    &\leq \frac{64r(st)^{3/2}\sqrt{\ln(4m \exp(\tfrac{\epsilon}{r}) / \delta)}}{\epsilon \sqrt{n}}.
\end{align*}
In the second, we have
\begin{gather*}
    \frac{e^{\alpha_0} + s}{e^{\alpha_0} - e^{(1-r)\alpha_0}} = \frac{1 + s / e^{\alpha_0}}{1 - e^{-r\alpha_0}} \leq 2\frac{1 + s e^{-\alpha_0}}{\min \{ 1, r \alpha_0\} } \\
    \frac{\sqrt{s + 2(e^{\alpha_0} - 1)}}{e^{\alpha_0} + (t-1)e^{(1-r)\alpha_0}-t} \leq \frac{\sqrt{2(s + e^{\alpha_0})}}{e^{\alpha_0}}.
\end{gather*}
This implies
\begin{align*}
    &\E[\emd(\tK, \tF)] \\
    &\;\leq 2\left(1 + \tfrac{1}{r\alpha_0}\right)(1 + se^{-\alpha_0})r \sqrt{\frac{st^3}{mn}} + 2\left( e^{-\alpha_0}\sqrt{s} + e^{-\alpha_0/2}\right)\sqrt{\frac{s^2t^2}{mn}} \\
    &\;\leq 2(1 + se^{-\alpha_0}) \sqrt{\frac{st^3}{mn}} + 2\left( e^{-\alpha_0}\sqrt{s} + e^{-\alpha_0/2}\right)\sqrt{\frac{s^2t^2}{mn}} \\
    &\;\leq 2(1 + \sqrt{s}e^{-\alpha_0/2}+ se^{-\alpha_0}) \sqrt{\frac{st^3}{mn}} \\
    &\;\leq 4(1 + se^{-\alpha_0}) \sqrt{\frac{st^3}{mn}}\\
    &\;\leq 4\sqrt{\frac{st^3}{mn}} + 1024\frac{r^2 \sqrt{m s^3t^3}}{\epsilon^2\sqrt{n}}\ln(4m \exp(\epsilon/r)/\delta) \\
    &\;\leq 4\sqrt{\frac{st^3}{mn}} + 32\frac{r \sqrt{s^3t^3}}{\epsilon\sqrt{n}}\sqrt{\ln(4m \exp(\epsilon/r)/\delta)}.
\end{align*}
In both cases, the desired bound has been shown.

\subsection{Proof of Lemma~\ref{lem:lap-freq}}
    We use the bound that $\emd(\tK, \tF) \leq \|\tK - \tF\|_1$. In each coordinate, the expected error introduced by the Laplace noise is at most $O(\frac{1}{n\epsilon})$, and thus $\E[\|\tK - \tF\|_1] \leq O(\frac{k}{n\epsilon})$. Normalizing will only reduce this error.

\subsection{Proof of Corollary~\ref{coro:krr-utility-central}}
\krracccent*
Our mechanism will simply combine the itemsets into one large itemset $K$ with $mn$ elements (and one global user), and then apply the algorithm of Theorem~\ref{thm:krr-utility}. By Theorem~\ref{thm:eps-amp}, the privacy budget is $(\alpha, \delta)$, where 
\begin{gather*}
    \alpha_0 = \begin{cases} \frac{\alpha \sqrt{n}}{32\sqrt {m \ln(4m e^{\alpha}/ \delta)}} & \text{ if $\alpha \sqrt{n} \leq 32 \sqrt{m \ln (4m e^{\alpha}/ \delta)}$} \\ 2\ln \left( \frac{\alpha \sqrt{n}}{16\sqrt{m  \ln(4m e^{\alpha} / \delta)}} \right) & 32 \sqrt{m \ln (4m e^{\alpha}/ \delta)} < \alpha \sqrt{n} < m \sqrt{n} \end{cases}
\end{gather*}
 Following the proof in Section~\ref{app:proof-krr-utility}, 
    (and setting $\alpha = \frac{\epsilon}{r}$), we can show that
    \[
        \E[\emd(\tK, \tF)] \leq 4\sqrt{\frac{st^3}{mn}} + 64\frac{r \sqrt{s^3t^3}}{\epsilon n}\sqrt{\ln(4m \exp(\epsilon/r)/\delta)}.
    \]

\end{document}